\documentclass[lettersize,onecolumn]{IEEEtran}

\usepackage{amssymb}
\usepackage{verbatim}
\usepackage{array}
\usepackage{latexsym}
\usepackage{enumerate}
\usepackage{amsmath}
\usepackage{amsfonts}
\usepackage{amsthm}
\usepackage{color}
\usepackage[english]{babel}
\usepackage{hyperref}
\usepackage{enumitem}
\usepackage{cleveref}
\crefname{subsection}{subsection}{subsections}
\Crefname{subsection}{Subsection}{Subsections}
\crefname{equation}{equation}{equations}

\usepackage{tikz}
\usepackage{wasysym}
\usepackage{mathrsfs}
\usepackage{setspace}
\usepackage{textcomp}
\usepackage{graphicx}
\usepackage{cite}
\usepackage{comment,cancel}
\graphicspath{ {images/} }
\usepackage{listings}
\usepackage[all]{xypic}

\theoremstyle{definition}
\newtheorem{theorem}{Theorem}[section]
\newtheorem{proposition}[theorem]{Proposition}
\newtheorem{lemma}[theorem]{Lemma}
\newtheorem{corollary}[theorem]{Corollary}

\newtheorem{definition}{Definition}
\newtheorem{example}[theorem]{Example}
\newtheorem{construction}[theorem]{Construction}

\def\F{\mathbb F}
\def\Q{\mathbb Q}

\def\Z{\mathbb Z}

\def\cC{\mathcal C}

\def\cO{\mathcal O}

\newcommand{\fq}{\mathbb {F}_q}
\newcommand{\fqq}{\mathbb {F}_{q^2}}

\def\End{\mbox{\rm End}}

\def\ord{\mbox{\rm ord}}
\def\deg{\mbox{\rm deg}}

\def\div{\mbox{\rm div}}
\def\min{{\rm min}}

\def\Div{\mbox{\rm div}}

\def\dim{\mbox{\rm dim}}
\def\supp{\mbox{\rm Supp}}
\def\fq{{\mathbb F}_q}

\def\supp{{\rm Supp}}

\newcommand\pl{\mathfrak{p}} 
\newcommand\ql{\mathfrak{q}} 
\newcommand\rl{\mathfrak{r}}

\newcommand{\mC}{{\mathcal{C}}}

\def\neg1{\text{\boldmath$1$}}

\def\neg1{\text{\boldmath$1$}}

\def\cC{\mathcal C}

\def\PP{\mathbb{P}}

\def\Fq{{\mathbb{F}_q}}
\def\Fqq{{\mathbb{F}_{q^2}}}
\def\Fqr{{\mathbb{F}_{q^r}}}

\DeclareMathOperator{\Hom}{Hom}

\newcommand{\ie}{\textit{i.e.,~}}

\usepackage{subcaption}

\definecolor{BC}{HTML}{0000aa}
\definecolor{GV}{HTML}{aa0000}

\newtheorem{remark}[theorem]{Remark}

\begin{document}
\title{Explicit and asymptotically good constructions of Algebraic Geometry codes in the sum-rank metric}
\author{Peter Beelen\thanks{Department of Applied Mathematics and Computer Science, Technical University of Denmark, DK-2800, Kongens Lyngby, Denmark, Email: pabe@dtu.dk} \and Elena Berardini\thanks{CNRS; IMB, University of Bordeaux, 33405 Talence, France, Email: elena.berardini@math.u-bordeaux.fr} \and Anina Gruica\thanks{Department of Applied Mathematics and Computer Science, Technical University of Denmark, DK-2800, Kongens Lyngby, Denmark, Email: anigr@dtu.dk} \and Maria Montanucci\thanks{Department of Applied Mathematics and Computer Science, Technical University of Denmark, DK-2800, Kongens Lyngby, Denmark, Email: marimo@dtu.dk}}

\maketitle

\begin{abstract}
Algebraic Geometry (AG) codes (\ie linear codes from algebraic function fields) in the Hamming metric were proposed by Goppa in 1980 and have been intensively studied ever since. Linearized Algebraic Geometry codes, the analogue of AG codes in the sum-rank metric, were instead introduced more recently \cite{BC2024}, using quotients of the ring of Ore polynomials with coefficients in an algebraic function field. In this paper we further investigate the results in \cite{BC2024}, providing explicit, optimal and asymptotic constructions. 
\end{abstract}


\section{Introduction}

Linear codes in the Hamming metric have been investigated for more than seventy years, and have played and still play a central role in correcting errors in noisy communication channels. Linear codes in the rank metric were introduced by Delsarte \cite{D1978} for purely combinatorial interest. It is, however, well known that rank metric codes offer a solution to the problem of error amplification in linear network coding, both in the one-shot and multi-shot regime \cite{KK2008,RK2018,SKK2008}.
In the latter scenario, using codes with the sum-rank metric can significantly reduce the size of the
network alphabet \cite{NU2010}, which is particularly handy in fast-evolving systems. 

Sum-rank metric codes were introduced more recently than Hamming and rank metric codes, and can be seen a generalization of them. They can in fact be defined as follows.

Let $\mathbb{F}_q$ be a finite field with $q=p^n$ elements ($p$ a prime number) and let $s \in \mathbb{Z}$ be a positive integer. Let further $\underline{V}=(V_1,\ldots,V_s)$ and $\underline{W}=(W_1,\ldots,W_s)$, where $V_i$ and $W_i$ are vector spaces over $\mathbb{F}_q$ for all $i\in \{1,\dots,s\}$. Denote with $n_i$ and $m_i$ the dimension of $V_i$ and $W_i$ over $\mathbb{F}_q$, respectively. Furthermore, let
$$\Hom_{\fq}(\underline{W},\underline{V})=\Hom_{\fq}(W_1,V_1) \times \Hom_{\mathbb{F}_q}(W_2,V_2) \times \cdots \times \Hom_{\mathbb{F}_q}(W_s,V_s),$$
where for $i \in \{1,\dots,s\}$, $\Hom_{\mathbb{F}_q}(W_i,V_i)$ denotes the ${\mathbb{F}_q}$-vector space of dimension $m_in_i$ of ${\mathbb{F}_q}$-linear homomorphism $W_i \rightarrow V_i$. 
\begin{definition}
Let $\underline{\varphi}=(\varphi_1,\ldots,\varphi_s) \in \Hom_{\mathbb{F}_q}(\underline{W},\underline{V})$. Then, the \textit{sum-rank weight} of $\underline{\varphi}$ is
$$w_{srk}(\underline{\varphi}):=\sum_{i=1}^{s}\mathrm{rk}(\varphi_i).$$
Given $\underline{\varphi},\underline{\psi} \in \Hom_{\mathbb{F}_q}(\underline{W},\underline{V})$ we define the sum-rank metric distance of $\underline{\varphi}$ and $\underline{\psi}$ as
$$d_{srk}(\underline{\varphi},\underline{\psi}):=w_{srk}(\underline{\varphi}-\underline{\psi}).$$
\end{definition}

With this notion of distance in mind we can define codes in the sum-rank metric as follows.

\begin{definition}
A \textit{sum-rank metric code} $\mC$ over ${\mathbb{F}_q}$ is an ${\mathbb{F}_q}$-linear subspace of $\Hom_{\mathbb{F}_q}(\underline{W},\underline{V})$. The parameters of $\mC$ are 
\begin{itemize}
    \item[(i)] \textit{length}: $n:=\dim_{\mathbb{F}_q} \Hom_{\mathbb{F}_q}(\underline{W},\underline{V})=\sum_{i=1}^s m_in_i$;
    \item[(ii)] \textit{dimension}: $k:=\dim_{\mathbb{F}_q} (\mC)$;
    \item[(iii)] \textit{minimum distance}: $d:=\min \{w_{srk}(\underline{c}) : \underline{c} \in \mC, \ \underline{c} \ne \underline{0}\}$.
\end{itemize}
\end{definition}

Note that when $n_i = m_i =1$ for all $i=1,\ldots,s$ the previous definition reduces to codes of length
$s$ with the Hamming metric and, in the case where $s = 1$, to rank metric codes. 

Apart from applications in multi-shot linear network coding \cite{MPK20191} and space-time coding, sum-rank metric codes can also be
used in distributed storage systems \cite{MPK20192}. Furthermore, convolutional codes endowed with the sum-rank metric have been considered in the literature, see \cite{MBK2016}, in order to address network streaming problems.

Despite its relatively recent introduction, several constructions of sum-rank metric codes are known. A general construction of MSRD codes (the analogue of MDS codes in the sum-rank metric) up to a certain length, along with a Welch--Berlekamp sum-rank decoding algorithm, can be found in \cite{MP2018} (see
also \cite{MPK20192}). These codes are linearized Reed--Solomon codes, that is, a hybrid between Reed--Solomon codes and Gabidulin codes, and their duals are again linearized Reed--Solomon codes \cite{MPK20191}. In \cite{MPK20192}, sum-rank alternant codes were introduced, while sum-rank BCH codes, Goppa codes and Reed--Muller codes have been
introduced in \cite{MP2021}, \cite{CD23} and \cite{BC2025}, respectively.

As will happen in this paper, in all the above-mentioned results, most of the sum-rank metric codes considered can be viewed as linear spaces over an extension field, which in particular requires that all matrix blocks (when representing each homomorphism as a matrix) have the same number of columns.
Constructions for which this is not the case exist; see, for example, \cite{MP2024, NSZ2023, MP2023}.

The analogue of Algebraic Geometry (AG) codes in the sum-rank metric, called linearized AG (LAG) codes, was introduced in \cite{BC2024}, using Ore polynomials with coefficients over a curve function field. Successively, further results were obtained in \cite{ZZ2025}, while the duality theory and a decoding algorithm for these codes were developed in \cite{BCD26}. 

In this paper, we will consider an important case of interest for sum-rank metric codes, which occurs when we are looking at a finite extension ${\mathbb{F}_{q^r}}$ of degree $r$ of $\mathbb{F}_q$, and we set $V_i={\mathbb{F}_{q^r}}$ for every $i\in\{1,\ldots,s\}$. In this case, $n_i = r$ for all $i$ and the ambient space
$\Hom_{\mathbb{F}_q}(W,{\mathbb{F}_{q^r}})$ is itself a vector space over ${\mathbb{F}_{q^r}}$. We are then more particularly interested in
${\mathbb{F}_{q^r}}$-linear codes which are, by definition, ${\mathbb{F}_{q^r}}$-linear subspaces $\mC \subseteq \Hom_{\mathbb{F}_q}(\underline{W},{\mathbb{F}_{q^r}})$. We can then
define ${\mathbb{F}_{q^r}}$-variants of the parameters of the code, namely the ${\mathbb{F}_{q^r}}$-length of $\mC$, $n_{r}=\sum_{i=1}^s m_i$, and the ${\mathbb{F}_{q^r}}$-dimension of $\mC$,  $k_{r} := \dim_{\mathbb{F}_{q^r}}\mC$. The minimum distance $d$ of $\mC$ stays unchanged. Those three main
parameters are related by the equivalent version of the Singleton bound in the sum-rank metric \cite[Prop.~34]{MP2018},
$$d + k_{r} \leq n_{r} + 1.$$ Codes with
parameters attaining this bound are called Maximum Sum-Rank Distance (MSRD).

In this paper, we further investigate the machinery proposed in \cite{BC2024} for AG sum-rank metric codes. More precisely in \Cref{sec:prelim}, we recall the basic notions on Ore polynomial rings and the construction from \cite{BC2024}. In \Cref{sec:parameters}, we analyze LAG codes further and give several results that allow for a more explicit construction of LAG codes in under some technical but not particularly restrictive assumptions. In the same section, we extend the bound for the parameters of the resulting AG sum-rank metric codes by allowing the corresponding evaluation map not to be injective (see \Cref{thm:exact_dimension}). In \Cref{sec:explicitconstructions}, we provide several explicit AG codes in the sum-rank metric with good parameters. In particular, we recover the construction of linearized Reed--Solomon codes (\Cref{constr:1}) and obtain another class of MSRD codes from Kummer extensions of rational function fields (\Cref{constr:Kummer}). In \Cref{sec:maxfunctfields}, we focus on a special class of algebraic function fields called maximal. These codes have a designed minimum distance (like the standard AG codes), with a penalty from being optimal upper bounded by $rg$, where $g$ is the genus of the function field. In \Cref{sec:asymptotic}, we provide new asymptotic results. A first improvement of the known Gilbert--Varshamov bound for sum-rank metric codes is provided in \Cref{thm:asymptotic3}. This result is followed by further improvements obtained using explicit optimal and good towers, both if $q$ is a square (\Cref{thm:asymptotic}) and if $q$ is not a square (\Cref{thm:asymptotic4}). Lastly, in \Cref{sec:conclusion}, we draw some conclusions and offer some ideas for further development of our work.


\section{Preliminaries} \label{sec:prelim}
In this section, we recall basic notions on Ore polynomial rings and introduce the objects which we shall use to summarize the construction of linearized Algebraic Geometry codes as introduced in~\cite{BC2024}.
\smallskip

For the rest of the paper, $\Fq$ will denote the finite field with $q$ elements. We consider two function fields $K$ and $L$ defined over $\Fq$. Moreover, we assume that $L/K$ is Galois with cyclic Galois group of order~$r$ generated by $\Phi$. We use the notation $N_{L/K}$ for the field norm.

\subsection{Ore polynomial rings and the algebra $D_{L,x}$} \label{subsec:ore}
For us to explain the algebraic geometry construction of sum-rank metric codes from~\cite{BC2024}, we need some preliminaries on Ore polynomial rings, which is the purpose of this subsection. We are going to work in our setting, that is, with $L/K$ being a function field extension. Note however that all the theory developed below applies unchanged to any field extension.

We denote by $L[T;\Phi]$ the ring of Ore polynomials with coefficients in $L$ in the variable $T$. The elements of this ring are the polynomials in $L[T]$, with usual addition but multiplication twisted with the following rule
$$T \cdot a=\Phi(a) \cdot T.$$

For $x \in K^*$, consider the algebra
$$D_{L,x}:=L[T;\Phi]/\langle T^r-x \rangle.$$
This inherits a ring structure from $L[T;\Phi]$, because every polynomial in $L[T;\Phi]$ commutes with $T^r-x$ (recall that $x \in K$ so $\Phi(x)=x$). In fact, $D_{L,x}$ can also be seen as a free module over $L$ of rank $r$ with basis given by $T^i+\langle T^r-x \rangle$ with $i=0,\ldots r-1$. For simplicity of notation we will write $f(T) \in D_{L,x}$ instead of the coset notation $f(T)+\langle T^r-x \rangle$ (\ie the basis written before is $\{1,T,\ldots,T^{r-1}\}$).

The following two ring homomorphisms on the quotient ring $D_{L,x}$ will be crucial for the code construction (see \cite[Sec.~II]{BC2024} for more details):
\begin{itemize}
\item Let $u \in L^*$ and $v:=N_{L/K}(u) \in K$. Then, there is a (well-defined) ring homomorphism
\begin{align*}
    \gamma_u:D_{L,x} &\rightarrow D_{L,v^{-1}x}\\ T &\mapsto uT
\end{align*}
    which is invertible, with inverse $\gamma_{u^{-1}}$. In particular, the rings $D_{L,x}$ and $D_{L,v^{-1}x}$ are isomorphic.
    \item When $x=1$, we have the following well-defined isomorphism:
    \begin{align*}
    \varepsilon: D_{L,1} &\rightarrow \End_K(L)\\ T &\mapsto\Phi
\end{align*}
\end{itemize}

Finally, for $f\in D_{L,x}$, consider the multiplication by $f$ $L$-linear map 
     $D_{L,x} \rightarrow D_{L,x}, g \mapsto gf.$ 
Its matrix form is
\begin{equation}\label{eq:Mf}
M_f:=\begin{pmatrix}
f_0 & x\cdot \Phi(f_{r-1}) & \ldots & x \cdot \Phi^{r-1}(f_1) \\
f_1 & \Phi(f_0) & \ldots & x \cdot \Phi^{r-1}(f_2) \\
\vdots & \vdots & & \vdots \\
f_{r-2} & \Phi(f_{r-3}) & \ldots & x \cdot \Phi^{r-1}(f_{r-1})\\
f_{r-1} & \Phi(f_{r-2}) &  \ldots & \Phi^{r-1}(f_0)
\end{pmatrix}.
\end{equation}

This map will play a role in the explicit construction of linearized AG codes later.

\subsection{Linearized AG codes}
We are now going to present AG codes in the sum-rank metric, termed linearized AG codes, as introduced in \cite{BC2024}.

We let $\Div_\Q(K)$ and $\Div_\Q(L)$ be the group of divisors on $K$ 
and $L$, respectively, where we allow the coefficients to be in $\Q$ (instead of only in $\Z$, in which case we will use the notations $\Div(K)$ and $\Div(L)$). 
To avoid confusion, we reserve the letter $\pl$ (resp.~$\ql$) to
denote places of $K$ (resp.~of $L$). We say that a place $\ql$ is \emph{above}
$\pl$, denoted as $\ql | \pl$, when $\ql \cap K=\pl$. We denote by $\PP_K$ and $\PP_L$ the set of places of $K$ and $L$, respectively, and set $\pi : 
\mathbb{P}_L\to \mathbb{P}_K$ to be the map sending $\ql\mapsto \pl=\ql\cap K$.

 We denote by $\nu_\pl(\cdot)$ and $\nu_{\ql}(\cdot)$ the valuations associated with $\pl$ and $\ql$, respectively. Since we are working with Galois extensions, the ramification indices $e_{\ql|\pl}$ of $\ql$ over $\pl$ are all equal for a fixed $\pl$, thus, we simply denote them by $e_\ql$ or $e_\pl$. 
 
Finally, for a place $\pl$ in $K$ we denote with $K_\pl$ the completion of $K$ at $\pl$, and for a place $\ql$ in $L$ we denote with $L_\ql$ the completion of $L$ at $\ql$. Note that $L_\ql/K_\pl$ is a field extension, so that one can consider the field norm $N_{L_\ql/K_\pl}$. We also set the following notation $$L_\pl:=K_\pl \otimes_K L \sim \prod_{i=1}^{m_\pl} L_{\ql_i},$$
where $m_\pl$ denotes the number of places $\ql$ of $L$ above $\pl$. Note that $m_\pl$ divides $r$, since $L/K$ is Galois extension.

We now fix a function $x \in K^\times$ and consider the algebra
$$D_{L,x}=L[T;\Phi]/\langle T^r-x \rangle$$
introduced in the previous section.
For a place $\pl \in K$, we set
\begin{align*}   \rho_{\pl}=\frac{e_\pl\cdot\nu_{\pl}(x)}{r}=\frac{\nu_{\ql}(x)}{r}=\frac{a_{\pl}}{b_{\pl}},
\end{align*}
where $b_p$ is a positive integer relatively prime to $a_\pl$. 
For the construction, we also fix:
\begin{enumerate}[label=(\roman*)] 
    \item a divisor $E=\sum_{\ql \in \mathbb{P}_L} n_{\ql} \ql \in \div_{\Q}(L)$, where for all $\ql$ we have $n_{\ql} \in \frac{1}{b_{\pl}}\Z$ where $\pl$ is the place of $K$ below $\ql$;
    \item a positive integer $s$ and $s$ rational places $\pl_1,\dots,\pl_s \in K$ which do not belong to $\pi(\supp(E))$.
\end{enumerate}

We will also need the following hypotheses:
\begin{itemize}
    \item[(\textbf{H1})] There are no non-zero zero divisors in $D_{L,x}$.
    \item[(\textbf{H2})] For all $i \in \{1,\dots,s\}$ it holds: for all places $\ql$ above $\pl_i$ there exists $u_{\ql} \in L_{\ql}^{\times}$ such that $\nu_{\ql}(u_{\ql})=\frac{e_{\pl_i}}{r}\cdot{\nu_{\pl_i}(x)}$ and 
    \begin{align*}
        x=\displaystyle\prod_{\ql | \pl_i}N_{L_{\ql}/K_{\pl_i}}(u_{\ql}).
    \end{align*}
\end{itemize}
We are now ready to describe the code construction. We consider the Riemann--Roch space of $D_{L,x}$ associated with the divisor $E$ as introduced in \cite[Def.~3 and Eq.~5]{BC2024}
$$\Lambda_{L,x}(E) = \bigoplus_{i=0}^{r-1} L(E_i) \cdot T^i,$$
where 
the divisors $E_i$, for $(0 \leq i < r)$, are defined by
$$E_i := \sum_{\ql \in \mathbb{P}_L} \big\lfloor
n_\ql + i \cdot \rho_{\pi(\ql)} \big\rfloor \cdot \ql\in \Div(L)$$
and the $L(E_i)$s are classical Riemann--Roch spaces of $L$.

For each $\pl$, with the notation of \Cref{sec:prelim}, we have an isomorphism
$$\varepsilon_\pl :
D_{L_\pl,x} \stackrel{\gamma_{u_\pl}}\longrightarrow
D_{L_\pl,1} \stackrel{\varepsilon}\longrightarrow
\End_{K_\pl}(L_\pl),$$
where $D_{L_\pl,x}=L_\pl[T;\Phi]/\langle T^r-x \rangle$ and $D_{L_\pl,1}=L_\pl[T;\Phi]/\langle T^r-1 \rangle$, respectively.
Following the discussion in \cite[Sec.~A]{BC2024} we see that restricting $\varepsilon_\pl$ to $\Lambda_{L,x}(E) $ and composing with the reduction modulo $t_\pl$, a uniformizer of $\pl$, we  get
$$\bar \varepsilon_\pl :
\Lambda_{L,x}(E) \stackrel{\varepsilon_\pl }\longrightarrow
\End_{K_\pl}(\cO_{L_\pl})\xrightarrow{\mathrm{mod}\,t_\pl}
\End_{\Fq}(V_\pl),$$
where we used the notation 
$V_\pl$ for the $r$-dimensional $\Fq$-algebra $\cO_{L_{\pl}}/t_\pl \cO_{L_{\pl}}$.
Finally, the linearized AG code $\cC(x; E; \pl_1, \ldots, \pl_s)$ is defined as the image of the multi-evaluation map
\begin{equation*}
\begin{array}{rcl}
\alpha : \quad 
\Lambda_{L,x}(E) & \longrightarrow & \prod_{i=1}^s \End_{\Fq}(V_{\pl_i}) \\
f & \mapsto & \big(\bar \varepsilon_{\pl_1}(f),\dots,\bar \varepsilon_{\pl_s}(f)\big).
\end{array}
\end{equation*}

\section{Linearized algebraic geometry codes and their parameters}\label{sec:parameters}

The goal of this section is twofold. First, we give a complete study of the parameters of linearized AG codes, extending what was done in \cite{BC2024}. Secondly, we prove some technical lemmas which, under some minor hypotheses on the evaluation places, allow for a simplified construction of the codes, and which we shall use for the explicit construction of \Cref{sec:explicitconstructions}. 

In this section, we draw the notations from \Cref{sec:prelim}, and we let $g_K$ and $g_L$ denote the genus of $K$ and $L$ respectively.

\subsection{The parameters of LAG codes} We want to study the parameters of the linearized AG (LAG) code $\cC:=\cC(x; E; \pl_1, \ldots, \pl_s)$ introduced at the end of the previous section.

The $\Fq$-linear code $\cC$ clearly has length $n=sr^2$. Bounds on the dimension and minimum distance of the LAG code were given in \cite[Thm.~2]{BC2024}, assuming that $\deg_L(E) < rs$, which in turn implies the injectivity of the map $\alpha$. In what follows, we prove the parameters of the codes without the hypothesis on $\deg_L(E)$.

\begin{theorem}\label{thm:exact_dimension}
We assume (\textbf{H1}) and (\textbf{H2}). The parameters of the code $\cC(x; E; \pl_1, \ldots, \pl_s)$ satisfy
\begin{align*}
k &=\dim_{\mathbb{F}_q}\, \Lambda_{L,x}(E) - \dim_{\mathbb{F}_q} \, \Lambda_{L,x}\left(E-\sum_{i=1}^s \sum_{\ql|\pl_i} e_{\pl_i}\ql\right),\\
d & \geq sr - \deg_L(E).
\end{align*}
\end{theorem}
\begin{proof}
The statement on the dimension follows once we show that $\ker \alpha = \Lambda_{L,x}(E-\sum_{i=1}^s \sum_{\ql|\pl_i} e_{\pl_i}\ql).$
This in turn follows if we can show that 
$\ker \bar \varepsilon_\pl=\Lambda_{L,x}(E- \sum_{\ql|\pl} e_\pl\ql)$ for all $\pl \in \{\pl_1,\dots,\pl_s\}.$ Let $f=\sum_{i=0}^{r-1}f_i T^i \in \ker \bar \varepsilon_\pl$. Since $\ker \bar \varepsilon_\pl \subseteq \Lambda_{L,x}(E)$, we only need to consider $\nu_\ql(f_i)$ for $\ql$ lying above $\pl \in \{\pl_1,\dots,\pl_s\}$. 

Now let $t$ be a local parameter for such $\pl$. Then $\gamma_{u_\pl}(f)=\sum_{i=0}^{r-1}f_i u_\pl \cdots \Phi^{i-1}(u_\pl) T^i$ and since $f \in \ker \bar \varepsilon_\pl$,
the matrix $M_{\gamma_{u_\pl}(f)}$, see equation \eqref{eq:Mf}, is in the kernel of the reduction modulo $t$ map.
Since $\nu_{\ql}(t)=e_\pl$, this implies that $\nu_\ql(f_i)+i e_\pl \nu_\pl(x)/r=\nu_\ql(f_i u_\pl\cdots \Phi^{i-1}(u_\pl)) \ge e_\pl$ for all $i$. Hence $f_i \in L(E_i)$ for all $i$, which implies that
$f \in \Lambda_{L,x}(E- \sum_{\ql|\pl} e_\pl\ql)$.

Conversely, assume that $f=\sum_{i=0}^{r-1}f_i T^i \in \Lambda_{L,x}(E- \sum_{\ql|\pl} e_\pl\ql)$. Then by definition of the space $\Lambda_{L,x}(E- \sum_{\ql|\pl} e_\pl\ql)$, we can conclude that $\nu_\ql(f_i) \ge e_\pl-i e_\pl \nu_\pl(x)/r.$ Therefore $f_iu_\pl\cdots \Phi^{i-1}(u_\pl) \equiv 0 \pmod{t}$, which implies that $f \in \ker \bar \varepsilon_\pl$. This concludes the proof.

Finally, the bound on the minimum distance $d$ was proven in \cite[Thm.~2]{BC2024}, assuming $\deg_L(E) < rs$. Clearly the bound also holds when $rs - \deg_L(E)\leq 0$.
\end{proof}

Whenever the evaluation map $\alpha$ is injective the dimension $k$ of $\cC(x; E; \pl_1, \ldots, \pl_s)$ satisfies
$k=\dim_\Fq \Lambda_{L,x}(E)$. This can be computed exactly via the Riemann--Roch theorem proved in \cite[Thm.~1.1.10]{BCD26}. However, in the following we will content ourselves with Riemann's inequality, which allows one to prove as in \cite[Thm.~2]{BC2024} that
\begin{equation}\label{ineq:dim}
k  \geq r \deg_L(E) - r(g_L - 1)-\frac {r^2} 2 \sum_{\pl \in \PP_K} \frac{b_\pl{-}1}{b_\pl e_\pl} \deg_K(\pl),  
\end{equation}
 when $K$ and $L$ have the same constant field, and 
\begin{equation} \label{ineq:dim2}
k  \geq r\deg_L(E) - r^2(g_K - 1)-\frac {r^2} 2 \sum_{\pl \in \mathbb{P}_K} \frac{b_\pl{-}1}{b_\pl} \deg_K(\pl),  
\end{equation}
when $L$ is a constant field extension of $K$ of degree $r$. 

\begin{remark}
In \cite{BC2024}, \cref{ineq:dim,ineq:dim2} are given in a unified way. There the language of curves is used, whereas here we use the language of function fields. In the setting of curves, a constant field extension gives rise to reducible curves, while in the function field setting this effect is not visible. This means that the notion of genus becomes different in the two settings.
\end{remark}

\begin{remark}\label{rem:any_dimension}
{Note that if the degree of $E$ is increased by $1$ (by adding a place of degree one), obtaining a new divisor $\tilde E$, then the dimension of $\Lambda_{L,x}(\tilde E)$ increases by at most $r$ over $\mathbb{F}_{q}$. This follows from \cite[Lem.~4]{BC2024}. In fact, applying this lemma, one has that 
 $$ \sum_{i=0}^{r-1} \deg_L(\tilde E_i)- \sum_{i=0}^{r-1} \deg_L(E_i) =\deg_L(\tilde E)-\deg_L(E)=r.$$
 Since $\deg_L(E_i)+1 \geq \deg_L(\tilde E_i)$ for all $i$, we get that necessarily $\deg_L(\tilde E_i)=\deg_L(E_i)+1$, and hence $\dim_{\mathbb{F}_q}\ L(E_i)+1 \geq \dim_{\mathbb{F}_q} \ L(\tilde E_i)$ for all $i$. This implies that 
 $\dim{\mathbb{F}_q}(\Lambda_{L,x}(\tilde E)) \leq  \dim_{\mathbb{F}_q}(\Lambda_{L,x}(E))+r$. In case $L/K$ is a constant field extension of degree $r$, this implies that the $\Fqr$-dimension of $\Lambda_{L,x}(\tilde E)$ can be at most one larger than that of $\Lambda_{L,x}(E)$.  
 Therefore, still assuming that $L=\Fqr K$, \Cref{thm:exact_dimension} implies that for any $k_r$ between $0$ and $sr$ there exists a choice of $E$ such that the code $\cC(x; E; \pl_1, \ldots, \pl_s)$ has $\Fqr$-dimension $k_r$.}
 \end{remark}

\begin{corollary}[\textnormal{see \cite[Cor.~2]{BC2024}}] \label{cor:singdef}
Assume (\textbf{H1}) and (\textbf{H2}), and that $\deg_L(E) < s r$. Moreover, assume that $L$ has full constant field $\fq$. Writing $n$, $k$ and $d$ for the length, the dimension and the
minimum distance of $\cC(x; E; \pl_1, \ldots, \pl_s)$, respectively, 
we have
$$r d + k \geq n + r - \left(r g_L+\frac {r^2} 2 \sum_{\pl \in \mathbb{P}_K} \frac{b_\pl{-}1}{b_\pl e_\pl} \deg_K(\pl)\right).$$
\end{corollary}

\begin{corollary}[\textnormal{see \cite[Sec.~IV.C]{BC2024}}] \label{cor:singdefisotrivial}
Assume (\textbf{H1}) and (\textbf{H2}), and that $\deg_L(E) < s r$. Assume further that $L/K$ is a constant field extension of degree $r$.  Writing $n_r$, $k_r$ and $d$ for the $\Fqr$-length, the $\Fqr$-dimension and the
minimum distance of $\cC(x; E; \pl_1, \ldots, \pl_s)$, respectively, 
we have
\begin{equation}\label{eq:boundisotrivialcase}
d + k_r \geq n_r + 1 - \left(r (g_K-1)+1
  + \frac {r} 2 \sum_{\pl \in \PP_K} \frac{b_\pl-1}{b_\pl} \deg_K(\pl)\right).
  \end{equation}
\end{corollary}

Note that if $$r g_L + \frac {r^2} 2 \sum_{\pl \in \PP_K} \frac{b_\pl{-}1}{b_\pl e_\pl} \deg_K(\pl)=0$$
in \Cref{cor:singdef}, then the $\Fq$-linear code $\cC(x; E; \pl_1, \ldots, \pl_s)$ is MSRD.

In the same way, if
$$r (g_K-1)+1
  + \frac {r} 2 \sum_{\pl \in \PP_K} \frac{b_\pl-1}{b_\pl} \deg_K(\pl)=0$$
in \Cref{cor:singdefisotrivial}, then the $\Fqr$-linear code $\cC(x; E; \pl_1, \ldots, \pl_s)$ is MSRD.

\subsection{Some technical lemmas towards explicit examples}\label{sec:technicallemmas}

To conclude this section, we present some lemmas we shall extensively use in the next section to present explicit examples of LAG codes.

First, to construct LAG codes, one needs to satisfy hypotheses (\textbf{H1}) and (\textbf{H2}). Since neither of them is straightforward to check, as in~\cite{BC2024} we give sufficient conditions to ensure both (\textbf{H1}) and (\textbf{H2}), that are easier to apply.

\begin{lemma}[\textnormal{see \cite[Lem.~5]{BC2024}}] \label{lem:H11}
If there exists a place $\pl \in \mathbb{P}_K$ which is inert in $L$ and at which $\nu_{\pl}(x)$ is coprime to $r$, then (\textbf{H1}) holds.
\end{lemma}

\begin{lemma}[\textnormal{see \cite[Cor.~30.7]{Rainer}}] \label{lem:H12}
If the smallest positive integer $d$ such that $x^d \in N_{L/K}(L)$ is equal to $r$, then (\textbf{H1}) holds.
\end{lemma}

\begin{lemma}[\textnormal{see \cite[Lem.~6]{BC2024}}] \label{lem:H2}
 If $\pl_i$ is unramified in $L$ and $\nu_{\pl_i}(x)$ is divisible by $r$ for all $i \in \{1,\dots,s\}$, then (\textbf{H2}) holds.
\end{lemma}

Secondly, one needs to compute the $u_\ql$ from hypothesis (\textbf{H2}). We point out that under some hypotheses on the evaluation places, the $u_\ql$ are not needed explicitly, only their reduction modulo $\ql$. Seen in this light, the following lemma is helpful.

\begin{lemma}\label{lem:constructionu}
Let $x \in K$ and suppose that $\nu_\pl(x)=0$ for some place $\pl$ of $K$. Let $L/K$ be an extension of function fields and $\ql|\pl$ for some place $\ql$ of $L$. Further, denote by $\kappa_\pl$ (resp. $\kappa_\ql$) the residue field of $\pl$ (resp. $\ql$) and by $K_\pl$ (resp. $L_\ql$) the completion of $K$ at $\pl$ (resp. $L$ at $\ql$). Finally, assume that $\ql|\pl$ is unramified in the extension $L/K$, \ie $e_\pl=1$. 

Let $\xi_\ql \in \kappa_\ql$ satisfies $N_{\kappa_\ql/\kappa_\pl}(\xi_\ql)=x(\pl)$, then there exists $u_\ql \in L_\ql$ such that $N_{L_\ql/K_\pl}(u_\ql)=x$ and $u_\ql(\ql)=\xi_\ql$. 
\end{lemma}
\begin{proof}
Let $U_K:=\{f \in K_\pl \, : \, \nu_{\pl}(f)=0\}$ and similarly $U_L:=\{g \in L_\ql \, : \, \nu_{\ql}(g)=0\}$. Since $N_{\kappa_\ql/\kappa_\pl}(\kappa_\ql)=\kappa_\pl$ and $\ql$ is unramified, we know that $N_{L_\ql/K_\pl}(U_L)=U_K$, see \cite[Chapt.~V, Prop.~3]{Serre_localfields}. Hence, since we assumed $\nu_\pl(x)=0$, there exists $u' \in L_\ql$ such that $N_{L_\ql/K_\pl}(u')=x$. Then $N_{\kappa_\ql/\kappa_\pl}(u'(\ql))=x(\pl)=N_{\kappa_\ql/\kappa_\pl}(\xi_\ql)$, so that there exists $\lambda \in \kappa_\ql$ such that $\lambda\cdot u'(\ql)=\xi_\ql$. Note that $N_{L_\ql/K_\pl}(\lambda)=N_{\kappa_\ql/\kappa_\pl}(\lambda)=1$. Hence we can choose $u_\ql=\lambda \cdot u'$.
\end{proof}

Using the previous lemma, we can show an easy way to describe a possible choice for the $u_\ql$ from hypothesis (\textbf{H2}).

\begin{lemma}
Let $\pl$ be a rational place of $K$ and assume that $\ql_1,\dots,\ql_t$ are all places of $L$ lying above $\pl$. Moreover, assume that $e_\pl=1$. If $\nu_{\pl}(x)=0$, then there exists $u_{\ql_1} \in L_{\ql_1}$ such that $N_{L_{\ql_1}/K_\pl}(u_{\ql_1})=x$. In particular, the $t$-tuple $(u_{\ql_1},1,\dots,1)$ satisfies the conditions in  (\textbf{H2}). 
\end{lemma}
\begin{proof}
This follows directly from \Cref{lem:constructionu}.
\end{proof}

\begin{remark}
From these lemmas and the way the codes are constructed, we see that it is enough to compute for each evaluation place $\pl$ and only one place $\ql$ of $L$ lying above $\pl$, a value $\xi_\ql$ from the finite field $\kappa_\ql$ satisfying $N_{\kappa_\ql/\kappa_\pl}(\xi_\ql)=x(\pl)$. Hence the lemmas simplify the code construction in two ways: in the first place one only needs to work with finite fields, not completions of the function fields and in the second place for each evaluation place $\pl$, only one value $\xi_\ql$ needs to be computed for some place $\ql$ of $L$ lying above $\pl$. For any other place $\ql$ of $L$ lying above $\pl$, one can simply choose $u_\ql=1$.  
\end{remark}

\section{Explicit sum-rank-metric AG codes}\label{sec:explicitconstructions}

\subsection{Constant field extension}

Let $q$ be a prime power, $K=\F_q(t)$ and $L=\F_{q^r}(t)$. We have $g_K=0$. The field $K$ has $q+1$ rational places, which are given by the places $\pl_{\alpha}=t-\alpha$, $\alpha \in \F_q$ and $\pl_{\infty}$. All these places are inert, thus in particular $e_{\pl_{\alpha}}=1$ for all $\alpha \in \F_q$ and $e_{\pl_{\infty}}=1$.  

\begin{construction}[AG linearized Reed--Solomon codes] \label{constr:1}
In this construction, we retrieve the linearized Reed--Solomon codes introduced in \cite{MP2018}. Note that this was already done in \cite{BC2024}, but here we provide all the details.

We fix the following:
\begin{itemize}
    \item[(i)] $x=t\in K$. Note that $\nu_{\pl_{0}}(x)=1$, $\nu_{\pl_{\infty}}(x)=-1$ and  $\nu_{\pl_i}(x)=0$ otherwise;
    \item[(ii)] the divisor $E=\frac{m}{r}\ql_{\infty}$. Note that for the place $\pl_{\infty}$ below $\ql_{\infty}$ we have $\rho_{\pl_{\infty}}=\frac{-1}{r}$ and so $b_{\pl_{\infty}}=r$;
    \item[(iii)] the $q-1$ distinct rational places $\pl_{\alpha_1},\dots,\pl_{\alpha_{q-1}} \in \mathbb{P}_K$, $\alpha_i \in \F_q^*$ for $i \in \{1,\dots,q-1\}$, which do not belong to $\pi(\supp(E))$.
\end{itemize}

For all $i \in \{1,\dots,q-1\}$, $\nu_{\pl_i}(x)$ is divisible by $r$, and so, by \Cref{lem:H2}, (\textbf{H2}) holds. Since $\pl_{\infty}$ is inert and $\nu_{\pl_{\infty}}(x)=-1$ is coprime to $r$, (\textbf{H1}) holds by \Cref{lem:H11}.

Now, we need to compute the $u_\pl$ from condition (\textbf{H2}) for all the places $\pl_{\alpha_i}$ indicated above.

Therefore, choose a rational place $\pl_\alpha$ of $K$ corresponding to $t-\alpha$, with $\alpha\in\F_q^*$. In the constant-field extension $L=\F_{q^r}(t)$, this place has a unique place $\ql_\alpha$ above it, corresponding again to the polynomial $t-\alpha$ but now viewed in $\F_{q^r}[t]$. We indicate with $x(\pl_\alpha)$ the residue of $x$ at the place $\pl_\alpha$. Since here $x=t$, this residue is $t\bmod(t-\alpha)=\alpha$. We therefore choose
$$
\gamma\in\F_{q^r}^*\qquad\text{with}\qquad N_{\F_{q^r}/\F_q}(\gamma)=\alpha.
$$
Hence, for the condition in (\textbf{H2}), using \Cref{lem:constructionu}, we can take the (length one) vector
$$
u_{\pl_\alpha}=(u_{\ql_\alpha}), \text{where }  u_{\ql_\alpha}\bmod \ql_\alpha=\gamma.
$$
 We now look at the Riemann--Roch space associated of~$D_{L,x}$ associated with $E=\frac{m}{r}\ql_\infty$:
\begin{align*}
    \Lambda_{L,x}(E)=\bigoplus_{i=0}^{r-1}L(E_i)T^i
\end{align*}
where 
\begin{align*}
    E_i:=\sum_{\ql \in \mathbb{P}_L}\lfloor n_{\ql}+i\rho_{\pi(\ql)} \rfloor \ql
\end{align*}
for $0 \le i < r$. We have $n_{\ql} = \frac{m}{r} \ne 0$ only for $\ql=\ql_{\infty}$. Moreover, $\rho_{\pi(\ql)} = 0$ unless $\pi(\ql)=\pl_{\infty}$ or $\pi(\ql)=\pl_{0}$. If $\pi(\ql)=\pl_{\infty}$ then $\rho_{\pi(\ql)}=\frac{-1}{r}$, and if $\pi(\ql)=\pl_{0}$, then $\rho_{\pi(\ql)} = \frac{1}{r}$. From this we obtain $$E_i=\left\lfloor \frac{m}{r}-\frac{i}{r}\right\rfloor \ql_{\infty}+\left\lfloor \frac{i}{r}\right\rfloor \ql_{0} = \left\lfloor \frac{m-i}{r}\right\rfloor \ql_{\infty}$$ for all $0 \le i < r$. We therefore have
\begin{align*}
    L(E_i)= \langle 1,t,\dots,t^{\lfloor (m-i)/r \rfloor}\rangle_{\F_{q^r}}
\end{align*}
for all $0 \le i < r$, and so
\begin{align*}
    \Lambda_{L,x}(E) &=
    \sum_{i=0}^{r-1}\langle 1,t,\ldots,t^{\lfloor (m-i)/r\rfloor}\rangle_{\F_{q^r}}T^i\\
    &=\langle 1,t,\dots,t^{\lfloor m/r \rfloor}\rangle_{\F_{q^r}} +\langle 1,t,\dots,t^{\lfloor (m-1)/r \rfloor}\rangle_{\F_{q^r}}T + \dots + \langle 1,t,\dots,t^{\lfloor (m-r+1)/r \rfloor}\rangle_{\F_{q^r}}T^{r-1}.
\end{align*}

Let $\Phi$ be the automorphism of $L$ determined by
$$
\Phi(\beta)=\beta^q\quad(\beta\in \F_{q^r}),\qquad \Phi(t)=t.
$$
An element $\beta\in\F_{q^r}$ is acted on by
$$
t\cdot\beta=\alpha\beta,
\qquad
T\cdot\beta=\gamma\,\Phi(\beta)=\gamma\beta^q.
$$
For $f \in \Lambda_{L,x}(E)$ we have $f=\sum_{r=1}^{r-1}f_i(t)T^i$ where $f_i(t) \in \langle 1,t,\dots, t^{\lfloor (m-i)/r\rfloor}\rangle_{\F_{q^r}}$. The remaining factor comes from iterating the action of $T$: since $T$ acts as $\beta\mapsto\gamma\beta^q$, one has
$$
T^i\cdot\beta=\gamma_i\beta^{q^i},\qquad
\gamma_i:=\gamma\,\Phi(\gamma)\cdots \Phi^{i-1}(\gamma),
$$
with $\gamma_0=1$. Hence the term $f_i(t)T^i$ acts after reduction as $\beta\mapsto f_i(\alpha)\gamma_i\beta^{q^i}$. For such an $f$, the corresponding rank-metric codeword is the $\F_q$-linear map $\overline\varepsilon_{\pl_\alpha}(f)$ described as
\begin{align*}
    \beta \longmapsto \sum_{i=0}^{r-1} f_i(\alpha)\gamma_i \beta^{q^i}, \qquad \gamma_i := \gamma \Phi(\gamma) \cdots \Phi^{i-1}(\gamma),
\end{align*}
where $\gamma_0=1$. Doing this for all the $q-1$ rational places fixed in the beginning explicitly gives the codeword in $\cC(x; E; \pl_{\alpha_1}, \ldots, \pl_{\alpha_{q-1}})$ associated with $f$. Note that this construction yields linearized Reed--Solomon codes. Furthermore, for $s=1$ and choosing $\alpha=\gamma=1$, this is precisely a Gabidulin code.

\smallskip
Finally, we turn to the parameters of the $\F_q$-linear code $\cC=\cC(x; E; \pl_{\alpha_1}, \ldots, \pl_{\alpha_{q-1}})$. 
We immediately get that the length of $\cC$ is $(q-1)r^2$. For the dimension, we can easily see from the shape of the Riemann--Roch space that we have $\dim_{\F_{q^r}}(\Lambda_{L,x}(E)) = \sum_{i=0}^{r-1} \left\lfloor \frac{m-i}{r}\right\rfloor$. We already know from~\cite{BC2024} that $\cC$ is MSRD, but for completeness, we compute the Singleton defect here, showing that indeed it is equal to $0$. Corollary \ref{cor:singdefisotrivial} implies that the Singleton defect is
\begin{align} \label{eq:singdef}
    r (g_K-1)+1
  + \frac {r} 2 \sum_{\pl \in \PP_K} \frac{b_\pl-1}{b_\pl} \deg_K(\pl).
\end{align}
Using that $b_{\pl}=r$ if $\pl=\pl_{0}$ or $\pl=\pl_{\infty}$ and $b_{\pl}=0$ otherwise,  \cref{eq:singdef} becomes
\begin{align*}
    -r+1+\frac{r}{2}\left(\frac{r-1}{r}
    +\frac{r-1}{r}\right)=0,
\end{align*}
showing that the Singleton defect is 0, and thus that $\cC$ is MSRD.

\end{construction}

One can increase the length by considering a different function $x$ in order to be able to use $\pl_0$ as an evaluation place. However, the resulting code is no longer MSRD, as shown in the following construction. 
\begin{construction} \label{constr:2}
 We let $\delta > 1$ be an integer coprime to $r$. We fix the following:
\begin{itemize}
    \item[(i)] $x={p(t)}\in K$ where $p$ is an $\Fq$-irreducible polynomials of degree $\delta$. Note that $\nu_{\pl_{\infty}}(x)=-\delta$ and  $\nu_{\pl_i}(x)=0$ for $i \in \{1,\dots,q\}$;
    \item[(ii)] the divisor $E=\frac{m}{r}\ql_{\infty}$ (note that for the place $\pl_{\infty}$ below $\ql_{\infty}$ we have $\rho_{\pl_{\infty}}=\frac{-\delta}{r}$ and so $b_{\pl_{\infty}}=r$);
    \item[(iii)] the $q$ distinct rational places $\pl_{\alpha_1},\dots,\pl_{\alpha_q} \in \mathbb{P}_K$, $\alpha_i \in \F_q$ for $i \in \{1,\dots,q\}$ (which do not belong to $\pi(\supp(E))$).
\end{itemize}
For all $i \in \{1,\dots,q\}$, $\nu_{\pl_i}(x)$ is divisible by $r$, and so, by \Cref{lem:H2}, (\textbf{H2}) holds. Since $\pl_{\infty}$ is inert and $\nu_{\pl_{\infty}}(x)=-\delta$ is coprime to $r$, (\textbf{H1}) holds by \Cref{lem:H11}.

We will denote the $\F_q$-linear code $\cC(x; E; \pl_{\alpha_1}, \ldots, \pl_{\alpha_q})$ by $\cC$. As in \Cref{constr:1}, we look at the Riemann--Roch space $\Lambda_{L,x}(E)$. We have $n_{\ql} = \frac{m}{r} \ne 0$ only for $\ql=\ql_{\infty}$. Moreover, $\rho_{\pi(\ql)} = 0$ unless $\pi(\ql)=\pl_{\infty}$ or $\pi(\ql)=\pl_{p}$ (the place corresponding to the polynomial $p(t)$). If $\pi(\ql)=\pl_{\infty}$ then $\rho_{\pi(\ql)}=\frac{-\delta}{r}$, and if $\pi(\ql)=\pl_{p}$, then $\rho_{\pi(\ql)} = \frac{1}{r}$. From this we obtain $$E_i=\left\lfloor \frac{m}{r}-\frac{i\delta}{r}\right\rfloor \ql_{\infty}+\left\lfloor \frac{i}{r}\right\rfloor \ql_{p} = \left\lfloor \frac{m-i}{r}\right\rfloor \ql_{\infty}$$ for all $0 \le i < r$. We therefore have
\begin{align*}
    L(E_i)= \langle 1,t,\dots,t^{\lfloor (m-i)/r \rfloor}\rangle_{\F_{q^r}}
\end{align*}
for all $0 \le i < r$, and so
\begin{align*}
    \Lambda_{L,x}(E) = \langle 1,t,\dots,t^{\lfloor m/r \rfloor}\rangle_{\F_{q^r}} +\langle 1,\dots,t^{\lfloor (m-1)/r \rfloor}\rangle_{\F_{q^r}}T + \dots + \langle 1,\dots,t^{\lfloor (m-r+1)/r \rfloor}\rangle_{\F_{q^r}}T^{r-1}.
\end{align*}
This gives that $\dim_{\F_{q^r}}(\Lambda_{L,x}(E)) = \sum_{i=0}^{r-1} \left\lfloor \frac{m-i}{r}\right\rfloor$. We will also compute the Singleton defect here to see how \emph{far} $\cC$ is from being an MSRD code.
We have that $b_{\pl}=r$ if $\pl=\pl_{p}$ or $\pl=\pl_{\infty}$ and $b_{\pl}=0$ otherwise. Therefore the Singleton defect of \cref{eq:singdef} becomes
\begin{align*}
    -r+1+\frac{r}{2}\left(\frac{r-1}{r}
    +\frac{r-1}{r}\delta\right)=(r-1)\left(\frac{\delta-1}{2} \right).
\end{align*}
\end{construction}

\subsection{Kummer extension}
In this subsection we will construct another family of MSRD codes. This time we do not use constant field extensions.
Let $K=\F_q(t)$ and $L=\F_{q}(y)$ with $y^r=t$ where $r | (q-1)$. We have that the genus is $g_L=0$. The Galois group of $L/K$ is generated by the automorphism $\Phi$ of $L$ defined by $\Phi(y)=\xi\cdot y$, where $\xi$ is a primitive $r$-th root of unity.

\begin{construction}\label{constr:Kummer}
We let $\zeta \in \F_q$ be a primitive element of $\F_q$. We fix the following:
\begin{itemize}
    \item[(i)] $x=\zeta\in K$;
    \item[(ii)] the divisor $E=m\ql_{\infty}$ (note that for the place $\pl_{\infty}$ below $\ql_{\infty}$ we have $\rho_{\pl_{\infty}}=\frac{0}{r}=0$ and so $b_{\pl_{\infty}}=1$);
    \item[(iii)] the $q-1$ distinct rational places $\pl_{\alpha_1},\dots,\pl_{\alpha_{q-1}} \in \mathbb{P}_K$, $\alpha_i \in \F_q^*$ for $i \in \{1,\dots,q-1\}$ (which do not belong to $\pi(\supp(E))$).
\end{itemize}

Note that for all $i \in \{1,\dots,q-1\}$, $\pl_{\alpha_i}$ is unramified and $\nu_{\pl_i}(x)=0$. In particular, $\nu_{\pl_i}(x)$ is divisible by $r$, and so, by \Cref{lem:H2}, (\textbf{H2}) holds. To check (\textbf{H1}), we need a bit more work than in the previous constructions. We will need the following lemma.

\begin{lemma} \label{lem:norm}
We have $N_{L/K}(L^*) \cap \F_q = \{\zeta^{ir} : 1 \le i \le q-1\}$.
\end{lemma}
\begin{proof}
Let $f \in L^*$ and suppose that $N_{L/K}(f) \in \F_q$. Since $L$ is a rational function field, we can write $f=p(y)/\tilde{p}(y)$, for certain polynomials $p(y),\tilde{p}(y) \in \mathbb{F}_q[y]$ such that $\gcd(p(y),\tilde{p}(y))=1$. Since $N_{L/K}(f)$ is a constant, we can conclude that $p(y)$ and $\tilde{p}(y)$ have the same degree. In particular, the place $Q_\infty$ of $L$ is neither a zero nor a pole of $f$.

Now let us write
$$p(y)=c\prod_{i}q_i(y)^{n_i} \quad \text{and} \quad \tilde{p}(y)=\tilde{c}\prod_{j}\tilde{q}_j(y)^{m_j},$$
with $c,\tilde{c} \in \mathbb{F}_q$ nonzero constants, $q_i(y),\tilde{q}_j(y)$ certain monic, irreducible polynomials and $n_i,m_j$ positive integers. The zeroes (resp. poles) of $f$ are precisely the places of $L$ corresponding to the irreducible polynomials $q_i(y)$ (resp.~$\tilde{q}_j(y)$). 

We claim that $N_{L/K}(f)=N_{L/K}(c/\tilde{c})$ and with prove this claim using induction on $d:=\deg( p(y))$. The claim is trivial if $d=0$ is zero, so let us assume that $d>0$. Since $N_{L/K}(f)$ is a constant, there exist $j$ and $\ell$ such that $\tilde{q}_j(y)$ and $ \Phi^\ell(q_1(y))$ are the same irreducible polynomial apart from multiplication by a scalar. More precisely, using that both $q_1(y)$ and $\tilde{q}_j(y)$ are assumed to be monic, we see that $\xi^{\ell \, \deg(q_1(t))}\cdot\tilde{q}_j(y)=q_1(\xi^\ell y)$
Moreover, using that $\gcd(p(y),\tilde{p}(y))=1$, we know that $\tilde{q}_j(y) \neq q_1(y)$. Hence we may assume that $1 \le \ell \le r-1$. 

Now note that
\begin{align*}
N_{L/K}(q_1(y))& =N_{L/K}(q_1(\xi^\ell y))=N_{L/K}(\xi^{\ell \, \deg(q_1(t))}\cdot\tilde{q}_j(y))\\ &= (\xi^{\ell \,  \deg(q_1(t))})^r\cdot N_{L/K}(\tilde{q}_j(y))=N_{L/K}(\tilde{q}_j(y)).
\end{align*}
Hence
$$N_{L/K}(f)=N_{L/K}\left(f \cdot \frac{\tilde{q}_j(y)}{q_1(y)}\right) \quad \text{and inductively} \quad N_{L/K}(f)=N_{L/K}\left(\frac{c}{\tilde{c}}\right),$$
thus proving the claim.
Since $N_{L/K}(\mathbb{F}_q^*)=\{a^r : a \in \mathbb{F}_q^*\}=\{\zeta^{ir} : 1 \le i \le q-1\}$, the lemma follows.
\end{proof}

Combining \Cref{lem:H12} with \Cref{lem:norm}, we see that (\textbf{H1}) holds.

We will denote the $\F_q$-linear code $\cC(x; E; \pl_{\alpha_1}, \ldots, \pl_{\alpha_{q-1}})$ by $\cC$. As in \Cref{constr:1}, we look at the Riemann--Roch space $\Lambda_{L,x}(E)$. We have $n_{\ql} = m \ne 0$ only for $\ql=\ql_{\infty}$. Moreover, $\rho_{\pi(\ql)} = 0$ for all places $\ql \in \mathbb{P}_L$. We have 
$$E_i=m\ql_{\infty}$$
for all $0 \le i < r$.
Therefore
\begin{align*}
    L(E_i)=\langle 1,y,\dots,y^m\rangle_{\F_{q}}.
\end{align*}
Thus
\begin{align*}
    \Lambda_{L,x}(E) = \langle 1,y,\dots,y^m\rangle_{\F_{q}}+\langle 1,\dots,y^m\rangle_{\F_{q}}T +\dots + \langle 1,\dots,y^m\rangle_{\F_{q}}T^{r-1}
\end{align*} 
and $\dim_{\F_{q}}(\Lambda_{L,x}(E)) = r(m+1)$. We again compute the Singleton defect from \cref{eq:singdef} for $\cC$. We have that $b_{\pl}=1$ for all places $\pl$ in $\mathbb{P}_K$. Therefore \cref{eq:singdef} becomes $0$. This shows that this code is MSRD.

\begin{remark}
The code above is a priori linear only over $\F_q$. However, the codes are modules over $\Lambda_{L,x}(0)$ as well, since $\Lambda_{L,x}(E)$ is a module over $\Lambda_{L,x}(0)$. Similarly as above, one can see that in the case we are considering here, one has
$$\Lambda_{L,x}(0)= \oplus_{i=0}^{r-1} \F_q T^i \cong \F_q[T]/\langle T^r-\zeta\rangle \cong \F_{q^r}.$$
Hence these codes are $\F_{q^r}$-linear as well.
\end{remark}
\end{construction}

We end this subsection by looking at an example of the codes coming from the Kummer extension. We only treat the case $s=1$, so the case where our resulting code is a rank-metric code achieving the Singleton bound, \ie an MRD code.

\begin{example}\label{ex:MRD_Kummer}
We fix $s=1$ and consider as only evaluation place $\pl$, the place corresponding to the polynomial $t-1$ in $K$. Since $y^r=t$, reducing modulo $t_\pl:=t-1$ gives $y^r=1$. As $r\mid(q-1)$, all $r$-th roots of unity lie in $\F_q$, so $\pl$ splits completely in $L/K$. The $r$ places above it are all rational, namely the places $\ql_j$  corresponding to $y-\xi^j$, with $0 \le j \le r-1$. Since $L_{\ql_j}=K_\pl$ for $0 \le j \le r-1$, we may choose $u_\pl=(u_{\ql_0},\dots,u_{\ql_{r-1}})=(1,\dots,1,\zeta)$.

Since all places $\ql_j$ are rational, the corresponding residue fields $\mathcal{O}_{L_{\ql_j}}/t_\pl \mathcal{O}_{L_{\ql_j}}$ are just $\F_q$. Thus we have
$$
\mathcal{O}_{L_{\pl}}/t_\pl\mathcal{O}_{L_{\pl}}=\prod_{j=0}^{r-1}\mathcal{O}_{L_{\ql_j}}/t_\pl \mathcal{O}_{L_{\ql_j}}\cong \F_q^r.
$$
We write an element of this product as
$$
\alpha=(\alpha_0,\alpha_1,\ldots,\alpha_{r-1})\in\F_q^r,
$$
where the $j$-th component corresponds to the place $\ql_j$. 
Since $\Phi(y-\xi^j)=\xi (y-\xi^{j-1})$ for all $0 \le j \le r-1$, we see that $\Phi(\ql_j)=\ql_{j-1}$ for $1 \le j \le r-1$ and $\Phi(\ql_0)=\ql_{r-1}$. In particular, using that $\Phi$ acts trivially on $\fq$, this implies that the induced action of $\Phi$ on $\mathcal{O}_{L_{\pl}}/t_\pl\mathcal{O}_{L_{\pl}}$ is 
$$
\Phi(\alpha_0,\alpha_1,\ldots,\alpha_{r-1})
=(\alpha_1,\alpha_2,\ldots,\alpha_{r-1},\alpha_0).
$$
By the choice of $u_\pl$, we see that the corresponding action of $T$, given by that of $u_\pi \Phi$, becomes 
$$
T(\alpha_0,\alpha_1,\ldots,\alpha_{r-1})
=(\alpha_1,\alpha_2,\ldots,\alpha_{r-1},\zeta\alpha_0).
$$
Choosing the standard basis of $\fq^r$, we obtain that $\alpha(T) \in \End_{\fq}(\fq^r)$ is represented by the matrix
$$A:=\begin{pmatrix}
0 & 1 & & \\
 \vdots & & \ddots & \\
0 & &  & 1\\
\zeta & 0 & \cdots & 0
\end{pmatrix}.$$
Note that this gives $T^r(\alpha_0,\ldots,\alpha_{r-1})=\zeta(\alpha_0,\ldots,\alpha_{r-1})$, as required by the relation $T^r=x$ and the equality $x=\zeta$.

Finally, $y$ acts by left multiplication on 
$\mathcal{O}_{L_{\pl}}/t_\pl\mathcal{O}_{L_{\pl}}$ and again using the standard basis of $\mathcal{O}_{L_{\pl}}/t_\pl\mathcal{O}_{L_{\pl}}=\fq^r$, we find that $\alpha(y)$ is represented by the matrix 
$$D:=\begin{pmatrix}
y(\ql_0) & & & \\
& y(\ql_1) & & \\
& & \ddots & \\
& & & y(\ql_{r-1})
\end{pmatrix}=\begin{pmatrix}
1 & & & \\
& \xi & & \\
& & \ddots & \\
& & & \xi^{r-1}
\end{pmatrix}.
$$
The MRD code $\cC(x; E; \pl)$ is therefore explicitly given as 
$$
\cC(x; E; \pl)=\left\langle D^aA^b : 0\le a\le m,\ 0\le b\le r-1\right\rangle_{\F_q}.
$$
\end{example}
\begin{remark}
The idea to use Kummer extensions to construct MRD codes was also used in~\cite{augot2014generalization}. However, there MRD codes were constructed over the rational function field $\fq(y)$ itself, not over $\fq$. These codes were then used to construct certain convolutional codes.
\end{remark}

\section{Explicit constructions using maximal function fields}\label{sec:maxfunctfields}

So far we have been constructing sum-rank-metric codes using rational function fields. In \cite[Rmk.~5]{BC2024} $\Fqr$-sum-rank metric codes were constructed from any given function field $K$ over $\fq$ by choosing $L=\Fqr K$. We start by giving a slight variation of the construction from \cite[Rmk.~5]{BC2024}, giving codes with a slightly smaller Singleton defect. In what follows, we denote by $N_1(K)$ the number of $\Fq$-rational places of $K$.

\begin{theorem}\label{thm:variationRemark4}
Let $K$ be a function field over $\fq$ and assume that $N_1(K)>2g_K$. Then there exists an $\Fqr$-linear sum-rank metric code with parameters $(n_r,k_r,d)$, where $n_r=r(N_1(K)-2g_K-2)$, $1 \le k_r \le n_r$ and 
\begin{equation*}
d + k_r \geq n_r + 1 -(2r-1)g_K.
\end{equation*}
\end{theorem}
\begin{proof}
Let $\pl_1$ and $\pl_2$ be two distinct rational places of $K$. The theorem of Riemann--Roch guarantees the existence of a function $x \in L((2g_K+1)\pl_1)$ such that $\nu_{\pl_2}(x)=1$. Considering $L=\Fqr K$ and the place $\pl_2$, we see that \Cref{lem:H11} applies. As evaluation places, we choose all $N_1(K)-2g_K-2$ rational places of $K$ distinct from $\pl_1$ and $\pl_2$ and any rational zeroes that $x$ may have. Further, let $\ql_1$ be the place of $L$ lying above $\pl_1$ and choose $E=m\ql_1$ such that the resulting sum-rank metric code has dimension $k_r$. As observed in \Cref{rem:any_dimension}, this is always possible. Then using \Cref{cor:singdefisotrivial} and the fact that $x$ has at most $2g_K+1$ many zeroes (including $\pl_2$), we obtain that
\begin{align*}
d + k_r & \geq n_r + 1 - \left(r (g_K-1)+1
  + \frac {r} 2 \sum_{\pl \in \mathbb{P}_K} \frac{b_\pl-1}{b_\pl} \deg_K(\pl)\right)\\
  & \geq n_r + 1 -(2r-1)g_K.
\end{align*}
\end{proof}

\begin{remark}
The improvement in \Cref{thm:variationRemark4} compared to Remark 4 from \cite{BC2024} stems from the fact that the number of zeroes of the function $x$ we use is at most $2g_K+1$, while in \cite{BC2024} a different $x$ is used with up to $2g_K+r-1$ zeroes. 
\end{remark}

For specific function fields, one can sometimes find a function $x \in L(h\pl_1)$ satisfying $\nu_{\pl_2}(x)=1$ with $h$ strictly smaller than $2g_K+1$ as in the proof of \Cref{thm:variationRemark4}. This results both in longer codes and in a smaller Singleton defect. This observation is particularly useful for maximal function fields, as we show now.

\begin{proposition}\label{prop:maxfunctfields}
Let $K$ be a maximal function field over $\fqq$. Then there exist $\F_{q^{2r}}$-linear sum-rank metric codes with $\F_{q^{2r}}$-length $r(q^2-q+2qg_K)$, $\F_{q^{2r}}$-dimension $k_r$ between $1$ and $n_r$ and minimum distance $d$ satisfying
$$d+k_r \ge n_r+1-rg_K-(q-1)\frac{r-1}{2}.$$
\end{proposition}
\begin{proof}
Let $\pl_1$ and $\pl_2$ be rational places of $K$. We claim that there exists $x \in K$ such that $(x)=-q\pl_1+\pl_2+D$, where $D$ is an effective divisor not containing $\pl_1$ nor $\pl_2$.

By the natural embedding theorem \cite[Thm.~10.22]{HKT2008}, $K$ can be realized as the function field of a curve embedded in a non-degenerate Hermitian variety using the linear system corresponding to the Riemann--Roch space $L((q+1)\pl_1)$. Moreover, this curve is non-singular \cite[Theorem 10.31]{HKT2008}. Now let $\pl_2$ be a second rational place of $K$. Since the curve embedded in the Hermitian variety is nonsingular, $L((q+1)\pl_1)$ contains a function $w$ such that $\nu_{\pl_2}(w)=1$. If $\nu_{\pl_1}(w)=-q$, we can choose $x=w$. Otherwise we distinguish two cases:

Case 1: $\nu_{\pl_1}(w)>-q$. Since $K$ is a maximal function field over $\fqq$, by \cite[Prop.~10.9]{HKT2008} there exists a function $z \in L((q+1)\pl_1)$ such that $\nu_{\pl_1}(z)=-q$. Adding a suitable constant to $z$ if necessary, we may also assume that $\nu_{\pl_1}(z)>0$. If $\nu_{\pl_1}(z)=1$, we can choose $x=z$, otherwise, we choose $x=w+z$.

Case 2: $\nu_{\pl_1}(w)=-(q+1)$. We claim that we can choose $x$ such that $\nu_{\pl_1}(x)=-q$. By the so-called Fundamental Equation \cite[Page xvii (ii)]{HKT2008}, there exists a function $y$ such that $(y)_K=(q+1)(\pl_2-\pl_1)$. Therefore we can choose a constant $c \in \fqq$ such that $\nu_{\pl_1}(w+cy) \ge -(q+1)$. If equality holds, we can choose $x=w+cy$, otherwise, we are back in Case 1. 

Note that $\pl_2$ is a inert place with $\nu_{\pl_2}(x)=1$, hence \textbf{(H1)} holds by \Cref{lem:H11}. Note also that \textbf{(H2)} is satisfied for all rational places different from $\pl_1$ at which $x$ does not vanish, and the latter are at most $q$. Further, $N_1(K)=q^2+1+2qg_K$, since $K$ is a maximal function field. The rest of the proof is now very similar to that of \Cref{thm:variationRemark4}.
\end{proof}

 \begin{remark}
In the proof of \Cref{prop:maxfunctfields} we showed that for any maximal function field $K$ and any two of its rational places $\pl_1$ and $\pl_2$, there exists a function $x \in K$ with pole order at $\pl_1$ equal to $q$ having a simple zero at $\pl_2$. 
Hence the pair $(q,-1)$ is in the two-point Weierstrass semigroup $H(\pl_1,\pl_2)$ defined in \cite{BT2006}.   
 \end{remark}
 
Now we consider other constructions that require a condition on $r$. For any function field $K$ with full constant field $\fq$, the group of divisor classes of degree zero of $K$ is finite. Hence, given two distinct rational places $\pl_1$ and $\pl_2$ of $K$, there exists an integer $m$ such that $m(\pl_1-\pl_2)$ is a principal divisor. The smallest positive such $m$ is just the order of the divisor of $\pl_1-\pl_2$ in the group of divisor classes of degree zero. We will denote this order by $\ord(\pl_1-\pl_2)$. It is also known as the period of the two-point Weierstrass semigroup $H(\pl_1,\pl_2)$, see \cite{BT2006}. The first construction of sum-rank metric codes uses this period.

\begin{theorem}\label{thm:construction_generalfields}
Let $K$ be a function field over $\fq$ and let $\pl_1$ and $\pl_2$ be two distinct rational places of $K$. Further let $r$ be a positive integer such that $\gcd(r,\ord(\pl_1-\pl_2))=1$. Then, there exist $\Fqr$-linear sum-rank metric codes with $\Fqr$-length $r(N_1(K)-2)$, $\Fqr$-dimension $k_r$ between $1$ and $n_r$ and minimum distance $d$ satisfying
$$d+k_r \ge n_r+1-rg_K.$$
\end{theorem}
\begin{proof}
Let $L=\Fqr K$. To construct the code, we choose $x\in K$ such that $(x)_K=\ord(\pl_1-\pl_2)(\pl_1-\pl_2)$, and all $N_1(K)-2$ rational places of $K$ distinct from $\pl_1$ and $\pl_2$ as evaluation places. Note that all the rational places are unramified and $\nu_{\pl_i}(x)=0$ for all $i\neq 1,2$, hence \textbf{(H2)} is satisfied by \Cref{lem:H2}. Note further that $\pl_1$ (equivalently, $\pl_2$) is inert and $\nu_{\pl_1}(x)=\ord(\pl_1-\pl_2)$, which is coprime with $r$ by hypothesis; hence \textbf{(H1)} holds by \Cref{lem:H11}. Denote by $\ql_2$ the place of $L$ lying above $\pl_2$. We set $E=m\ql_2$, with $m$ chosen such that the dimension of the resulting code is $k_r$. Then, according to \Cref{cor:singdefisotrivial}, the code has minimum distance satisfying
$$d+k_r \ge n_r-r(g_K-1)-2\frac{r-1}{2}=n_r+1-rg_K.$$
\end{proof}

\begin{corollary}
Let $K$ be a maximal function field over $\fqq$. Further, let $r$ be a positive integer such that $\gcd(r,q+1)=1$. Then there exist $\F_{q^{2r}}$-linear sum-rank metric codes with $\F_{q^{2r}}$-length $r(q^2-1+2qg_K)$, $\F_{q^{2r}}$-dimension $k_r$ between $1$ and $n_r$ and minimum distance $d$ satisfying
$$d+k_r \ge n_r+1-rg_K.$$
\end{corollary}
\begin{proof}
For a maximal function field $K$ over $\Fqq$ it is known, see \cite[Eq.~10.8]{HKT2008}, that the order of $\pl_1-\pl_2$ divides $q+1$ for any two distinct rational places $\pl_1$ and $\pl_2$ of $K$. The corollary now follows directly from  \Cref{thm:construction_generalfields}.
\end{proof}

\section{Asymptotic results}\label{sec:asymptotic}

In \cite[Thm.~4]{BC2024}, it was shown that in case $q>9$ is a square, it is possible to construct families of $\Fqr$-linearized AG codes with length tending to infinity, whose rate $R$ and relative distance $\delta$ respect
\begin{equation}\label{eq:BCbound}
R+\delta > 1 - \frac{2}{\sqrt{q}-3}+\frac{1}{r(\sqrt{q}-3)}.\end{equation}
Note that the result stated in \cite{BC2024} is slightly different, because of a small calculation error there. 
Using the aforementioned bound, it was shown in \cite{BC2024} that there exist families of linearized AG codes which, in a certain range, beat the sum-rank version of the Gilbert--Varshamov bound (see \Cref{thm:srkGV} below), over finite fields $\Fq$ with $q$ an even power of a prime, $q\geq 11^2$, and for any $r$.

In this section we refine this asymptotic result in various ways. In the first place, we revisit the proof in \cite[Thm.~4]{BC2024} and improve the bound in equation \eqref{eq:BCbound}. Afterwards, we show that by choosing certain explicit towers of function fields for the construction of the family of linearized AG codes, one can improve the bound of equation \eqref{eq:BCbound} further in case $q$ is not a prime. This, in turn, enables us to enlarge the range of $R$, $\delta$ and $q$ over which we improve on the sum-rank version of the Gilbert--Varshamov bound. 

Before starting, we recall here the sum-rank version of the Gilbert--Varshamov bound \cite[Thm.~7]{OPB21}.
\begin{theorem}[Asymptotic Gilbert--Varshamov bound]\label{thm:srkGV}
Consider the finite field extension $\Fqr/\Fq$. 
For any positive integer $s$ and any real numbers $R, \delta \in (0,1)$ with $\delta > \frac{2}{rs}$ and
\begin{align*}
R \leq \; \delta^2-\delta\left(2+\frac{2}{sr}\right)+1+\frac{2}{sr}+\frac{1}{s^2r^2}- \frac{\sum^{\delta sr-1}_{i=1} \log_q(1+\frac{s-1}{i})+\log_q(\delta sr-1)}{r^2s}-\frac{\log_q(\gamma_q)}{r^2},\end{align*}
where $\gamma_q=\prod_{i=1}^\infty (1-q^{-i})^{-1}$, there exists a sum-rank metric codes in $\End_{\Fq}(\Fqr)^s$ of rate at least $R$ and relative minimum distance at least $\delta$.
\end{theorem}

\subsection{First improvements}\label{subsec:firstimprovements}

In this section, we show how to improve the bound given in equation \eqref{eq:BCbound} from \cite{BC2024} in general as well as extend it to nonprime, nonsquare finite fields. In other words, we assume that the cardinality of the finite field $\Fq$ is a proper prime power, say $q=p^m$ with $m>1$. As usual, we denote Ihara's constant by $A(q)$. The Drinfeld--Vladut bound states that $A(q) \le \sqrt{q}-1$. In \cite{BBGS15}, it was shown that
\begin{equation}\label{eq:non-square bound}
A(p^m)\ge 2\left(\frac{1}{p^{\lfloor m/2 \rfloor}-1}+\frac{1}{p^{\lceil m/2 \rceil}-1}\right)^{-1}.
\end{equation}
For convenience, we denote the right-hand side of this inequality by $H(p^m)$. Note that if $m$ is even, this bound simplifies to $A(q) \ge \sqrt{q}-1$. Note also that this bound together with the Drinfeld--Vladut bound implies the well-known result that $A(q)=\sqrt{q}-1$ in case $q$ is a square.

For the Hamming metric, \cref{eq:non-square bound} has been used to construct codes better than the Gilbert--Varshamov bound whenever $q$ is not a prime, $q \ge 49$ and $q \neq 125$ \cite{BBGS14}. In this section we will investigate what can be done in the sum-rank metric case. 

First of all, following the same approach as in the proof in \cite[Thm.~4]{BC2024}, one immediately obtains the following result. For the convenience of the reader and because we will revisit the proof several times later on, we give the main ingredients of the proof. For full details, see \cite{BC2024}.

\begin{theorem}\label{thm:asymptotic2}
Let $q$ be a prime power and suppose that $A(q)>2$. For all $R$ and $\delta$ in the interval $(0,1)$ such that 
\[R+\delta < 1- \frac{2r-1}{r}\frac{1}{A(q)-2},\]
there exists an asymptotic family of $\mathbb{F}_{q^r}$-linearized AG codes with rate at least $R$ and relative minimum distance at least $\delta$.
\end{theorem}
\begin{proof}
As in the proof of Theorem 4 from \cite{BC2024}, the main ingredient is to use a good family of function fields and then to consider the isotrivial case of the construction from \cite{BC2024} for each function field in the tower. More precisely, let $\mathcal K=(K_1,K_2,\dots)$ be a family of function fields over $\fq$ with limit $A(q)$, meaning that $\lim_{i \to \infty} N_1(K_i)/g_{K_i} = A(q)$. Here we denoted by $g_{K_i}$ the genus of $K_i$ and by $N_1(K_i)$ the number of rational places of $K_i$. Each function field in this tower has full constant field $\Fq$. Let $L_i=\Fqr K_i$, \ie the function field obtained from $K_i$ by extending the constant field to $\Fqr$. Since $\Fq$ is the full constant field of $K_i$, we see that $L_i/K_i$ is an unramified, cyclic Galois extension of degree $r$. 

Now let $r \ge 1$ be an integer and let $\pl$ be a rational place of $K_i$. As explained in \cite{BC2024}, there exists a function ${x} \in K_i$ with pole divisor $({x})_{\infty}=h\pl$, where $2g_{K_i} \le h \le 2g_{K_i}+r-1$ and $\gcd(r,h)=1$. The function ${x}$ has $h$ zeroes counted with multiplicity, since its pole divisor has degree $h$. These zeroes do not need to be rational, but might be, reducing the potential length of the code, since zeroes (and poles) of $x$ cannot be used as evaluation places. Using the construction from \cite[Rmk.~5]{BC2024} gives for any given $i \ge 1$, an $\Fqr$-linear sum-rank metric codes with parameters $(n_r,k_r,d)$ with $\Fqr$-length $n_r=r(N_1(K_i)-2g_{K_i}-r)$, $\Fqr$-dimension $0 \le k_r \le n_r$ and minimum distance $d$. By \cref{eq:boundisotrivialcase}, $d$ satisfies: 
\begin{align*}
d + k_r 
  &\ge n_r  + 1 - \left(r (g_{K_i}-1)+1
  + \frac {r-1} 2 (2g_{K_i}+r)\right)\\
&= n_r -(2r-1)g_{K_i}-\frac{(r-3)r}{2}.
\end{align*}
Here we used the fact that $b_\pl \le r$ and that $b_\pl=1$ whenever $\nu_\pl(x)=0$. Now dividing by $n_r$ and letting $i$ tend to infinity, the result follows provided $n_r=r(N_1(K_i)-2g_{K_i}-r)$ also tends to infinity. However, this is automatic if $A(q) >2$.
\end{proof}

\begin{remark}
If $q$ is a square, then $A(q)=\sqrt{q}-1$ and  \Cref{thm:asymptotic2} simplifies to \cref{eq:BCbound}. If $q$ is not a square, the precise value of $A(q)$ is not known, but if $q$ is not a square and not a prime number, the bound $A(q) \ge H(q)$ from \cref{eq:non-square bound} can be used to construct asymptotically good families of codes whenever 
\[R+\delta < 1- \frac{2r-1}{r}\frac{1}{H(q)-2}.\]
Note that $H(q)>2$ for $q>9$. 
\end{remark}

Next we consider an improvement of \Cref{thm:asymptotic2}.
\begin{theorem}\label{thm:asymptotic3}
Let $q$ be a prime power. For all $R$ and $\delta$ in the interval $(0,1)$ such that 
\[R+\delta < 1- \frac{2r-1}{r} \cdot \frac{q}{(q-1)A(q)},\]
there exists an asymptotic family of $\mathbb{F}_{q^r}$-linearized AG codes with rate at least $R$ and relative minimum distance at least $\delta$.
\end{theorem}
\begin{proof}
To prove this theorem, we modify the proof of \Cref{thm:asymptotic2} by choosing the function $x$ differently. Denote by $\tilde{x}$ the function denoted by $x$ in the proof. For any $\alpha \in \Fq$, the function $\tilde{x}-\alpha$ has the same pole divisor $(\tilde{x}-\alpha)_{\infty}=h\pl$. On the other hand, by the pigeonhole principle, there exists a choice of $\alpha$ for which at most $\lfloor N_1(K_i)/q \rfloor$ many zeroes of $\tilde{x}-\alpha$ are rational places. Now let $x:=\tilde{x}-\alpha$ for such a value of $\alpha$. Now we can proceed with the proof exactly as in the proof of \Cref{thm:asymptotic2}, except that this time the $\Fqr$-length of the code satisfies 
$$n_r \ge r\left(N_1(K_i)-\frac{N_1(K_i)}{q}\right)=rN_1(K_i)\frac{q-1}{q}.$$ 
In particular, the length of the codes will tend to infinity as $i$ tends to infinity, without any restriction on $q$.
\end{proof}

\begin{remark}
Note that \Cref{thm:asymptotic3} indeed gives a better result than \Cref{thm:asymptotic2}. This amounts to showing that $A(q) < 2q$, which follows directly from the Drinfeld--Vladut bound.
\end{remark}

\begin{remark}
Using \Cref{thm:asymptotic3} and the bound $A(q) \ge H(q)$ for non-prime $q$, one can conclude that one can construct families of sum-rank metric codes with rate tending to $R$ and relative minimum distance tending to $\delta$ whenever
$$R+\delta < 1- \frac{2r-1}{r} \cdot \frac{q}{(q-1)H(q)}.$$

In general, it is known that $A(q)>0$. If a tower or family of function fields over $\fq$ has limit $\lambda>0$, one can similarly as above construct families of sum-rank metric codes whenever
$$R+\delta < 1- \frac{2r-1}{r} \cdot \frac{q}{(q-1)\lambda}.$$
Therefore, families of asymptotically good sum-rank metric codes can be constructed for prime values of $q$ too. However, even using the best known lower bounds for $A(q)$ in case $q$ is a prime, these families of codes are not better than the Gilbert--Varshamov bound.
\end{remark}

\subsection{Further improvements}
The results in the previous subsection were very general and only used families of good function fields. In fact, only the limit of such families was used. Using the explicit description of some optimal and good towers of function fields, further improvements are possible in case $q$ is not a prime. In this subsection, we consider two such improvements: one where $q$ is a square, and one where $q$ is not a square.

\subsubsection{The case $q$ is a square} 
Let us first assume that the cardinality of the finite field is a square, say $q=\ell^2$, with $\ell$ a power of some prime. It is well known that there exist towers of function fields attaining the Drinfeld--Vladut bound. We will use the recursively defined tower $\mathcal F=(F_1,F_2,\dots)$ from \cite{GS96}, where
\[
F_i=\begin{cases}
\mathbb{F}_{q}(x_1) & \text{if $i=1$}\\
F_{i-1}(x_i) & \text{if $i>1$}
\end{cases},
\]
where $$x_i^\ell+x_i=\frac{x_{i-1}^\ell}{x_{i-1}^{\ell-1}+1} \quad \text{for $i>1$}.$$
Denoting by $g_{F_i}$ (resp. $N_1(F_i)$) the genus (resp. number of rational places) of $F_i$, it is known, \cite[Rmk.~3.8]{GS96} and \cite[Lem.~3.9]{GS96}. that 
\begin{equation}\label{eq:genusformula}
g_{F_i}=(\ell^{\lceil i/2 \rceil}-1)(\ell^{\lfloor i/2 \rfloor}-1) \quad \text{and} \quad N_1(F_i) \ge (\ell-1)\ell^i.
\end{equation}
Therefore
$$\lim_{i \to \infty} \frac{N_1(F_i)}{g_{F_i}} \ge \ell-1,$$
thus providing an example of an asymptotically optimal tower (\ie a tower of function fields attaining the Drinfeld--Vladut bound). In fact using the Drinfeld--Vladut bound, we see that $\lim_{i \to \infty} N_1(F_i)/g_{F_i} = \ell-1$. In particular, this implies that $\lim_{i \to \infty} (\ell-1)\ell^{i}/N_1(F_i)=1.$

More precisely, it is shown in \cite[Lem.~3.9]{GS96} that if $\alpha^\ell+\alpha \neq 0$ and $\alpha \in \fq$, then the zero of $x_1-\alpha$ splits completely in the extension $F_i/F_1$ for any $i \ge 1$. This explains why $N_1(F_i) \ge (\ell-1)\ell^i$. If a place of $F_1$ ramifies in the extension $F_i/F_1$, it is either the pole of $x_1$ or a zero of $x_1-\alpha$ for some $\alpha \in \fq$ satisfying $\alpha^\ell+\alpha=0$ \cite[Lem.~3.3]{GS96}. In fact the pole of $x_1$ 
is totally ramified in the extension $F_i/F_1$ for any $i \ge 1$ \cite[Lem.~3.3]{GS96}. This identifies a unique rational place of $F_i$, which we denote by $\pl_\infty^{(i)}$. Similarly, the zero of $x_i$ is totally ramified in the extension $F_i/\fqq(x_i)$, giving rise to a unique rational place of $F_i$, which we will denote by $\pl_0^{(i)}$. It is also a zero of $x_1$ and it holds that $\nu_{\pl_0^{(i)}}(x_1)=1$, see the discussion after Lemma 3.3 in \cite{GS96}.

\begin{theorem}\label{thm:asymptotic}
Let $q$ be a square and $r \ge 1$ an integer. For all $R$ and $\delta$ in the interval $(0,1)$ such that 
\[R+\delta < 1-\frac{1}{\sqrt{q}-1}
,\]
there exists an asymptotic family of $\mathbb{F}_{q^r}$-linearized AG codes with rate at least $R$ and relative minimum distance at least $\delta$.
\end{theorem}
\begin{proof}
Let $\mathcal F=(F_1,F_2,\dots)$ be the optimal tower from \cite{GS96} some of whose properties were paraphrased in the beginning of this subsection. To construct a family of codes, we choose for each $i$, the function $x:=x_1 \in F_i$. Note that $\nu_{\pl_0^{(i)}}(x)=1$ so that \Cref{lem:H11} applies. Avoiding the pole and zeroes of $x$ as evaluation places, we see that we can create $\Fqr$-linear sum-rank metric codes $(n_r,k_r,d)$ of $\Fqr$-length $n_r=r(\ell-1)\ell^{i-1}$, $\Fqr$-dimension
$0 \le k_r \le n_r$ and where $d$ satisfies the inequality from \Cref{cor:singdefisotrivial}. 
To estimate $d$ better, we study the summation on the right-hand side of this inequality in more detail. Since $b_\pl=1$ whenever $\nu_\pl(x)=0$, we will estimate how many zeroes and poles $x$ has in $F_i$. First of all, $x$ has only one pole, namely $\pl_\infty^{(i)}$. Further, any ramified place of $F_i$ in the extension $F_i/F_1$ is either $\pl_\infty^{(i)}$ or a zero of the function $x_1^\ell+x_1$. For this reason, the set $V$ consisting of the zeroes and poles of $x_1^\ell+x_1$ is called the ramification locus of the tower $\mathcal F$. 

The number of zeroes and poles of $x$ in $F_i$ will therefore be bounded by the number of places of $F_i$ lying above a place in $V$. Moreover, the tower $\mathcal F$ is weakly ramified, meaning that for any place $\pl$ of $F_i$ lying above some place $\rl$ of $F_1$ it holds that $d_{\pl|\rl}=2(e_{\pl|\rl}-1)$, where $d_{\pl|\rl}$ denotes the different exponent of $\pl$ and $e(\pl|\rl)$ its ramification index \cite[Lem.~7.4.6.]{Stich2}. Finally, note that for any place $\pl$ of $F_i$ lying above a place $\rl$ in $V$, one has $f_{\pl|\rl}=\deg_{F_i}(\pl)$, since $\deg_{F_1}(\rl)=1$ for any $\rl \in V$. Here $f_{\pl|\rl}$ denotes the relative degree of $\pl$ in the extension $F_i/F_1$. Using the Riemann--Hurwitz formula, one then obtains that:
\begin{align*}
2g_{F_i}-2 & = -2\ell^{i-1}+\sum_{\rl \in V} \sum_{\substack{\pl|\rl \\ \pl \in \mathbb{P}_{F_i}}} 2(e_{\pl|\rl}-1)\deg_{F_i}(\pl)\\
 &= -2\ell^{i-1}+2\sum_{\rl \in V} \sum_{\substack{\pl|\rl \\ \pl \in \mathbb{P}_{F_i}}} e_{\pl|\rl}f_{\pl|\rl}-2\sum_{\rl \in V} \sum_{\substack{\pl|\rl \\ \pl \in \mathbb{P}_{F_i}}} \deg_{F_i}(\pl)\\
 &=-2\ell^{i-1}+2\#V \cdot \ell^{i-1}-2\sum_{\rl \in V} \sum_{\substack{\pl|\rl \\ \pl \in \mathbb{P}_{F_i}}} \deg_{F_i}(\pl)\\
 &= 2\ell^i-2\sum_{\rl \in V} \sum_{\substack{\pl|\rl \\ \pl \in \mathbb{P}_{F_i}}} \deg_{F_i}(\pl).
\end{align*}
Hence using equation \eqref{eq:genusformula}, we find that
\begin{align*}
\sum_{\rl \in V} \sum_{\substack{\pl|\rl \\ \pl \in \mathbb{P}_{F_i}}} \deg_{F_i}(\pl) = {\ell^{\lfloor \frac{i}{2} \rfloor}+\ell^{\lceil \frac{i}{2} \rceil}}.
\end{align*}
Using this and the aforementioned fact that $b_\pl=1$ whenever $\nu_\pl(x)=0$, we find that
\begin{align*}
\sum_{\pl \in \mathbb{P}_{F_i}} \frac{b_\pl-1}{b_\pl} \deg_{F_i}(\pl) &\le \frac{r-1}{r}\sum_{\rl \in V} \sum_{\substack{\pl|\rl \\ \pl \in \mathbb{P}_{F_i}}}\deg_{F_i}(\pl)
 = \frac{r-1}{r} \left( {\ell^{\lfloor \frac{i}{2} \rfloor}+\ell^{\lceil \frac{i}{2} \rceil}} \right).
\end{align*}
But then the minimum distance of the sum-rank metric codes we constructed satisfies:
\begin{align*}
d+k_r &\ge n_r+1 -\left(r(g_{F_i}-1)+1+\frac{r-1}2 \left({\ell^{\lfloor \frac{i}{2} \rfloor}+\ell^{\lceil \frac{i}{2} \rceil}}\right) \right).
\end{align*}
Since $n_r=r(\ell-1)\ell^{i-1}$ and $\lim_{i \to \infty} (\ell-1)\ell^{i}/g_{F_i}=\sqrt{q}-1$, letting $i$ tend to infinity yields that we can construct asymptotic families of codes whose rate $R$ and relative minimum distance $\delta$ satisfy
\[R+\delta \ge 1-\frac{1}{\sqrt{q}-1}
\]
The theorem now follows.
\end{proof}

In \Cref{fig:GVevencase}, we picture the comparison between the GV bound, the bound from  Theorem 4 of \cite{BC2024} and our bounds from \Cref{thm:asymptotic}. We observe that, for $q=7^2$, there is a range where our bound improves on the Gilbert--Varshamov bound, for any $r$. This matches the situation in the Hamming case. 

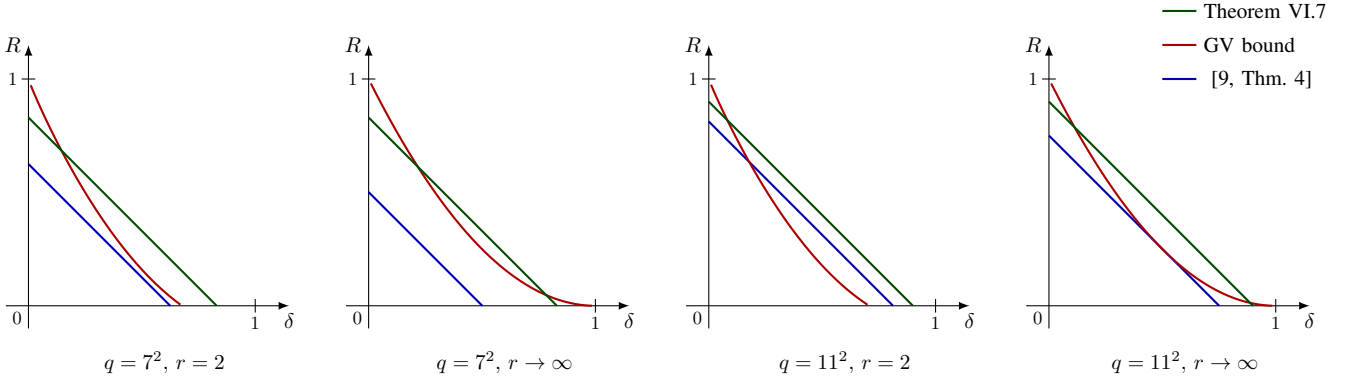
\begin{figure}
\begin{tikzpicture}[scale=3]
\begin{scope}[]
\draw[-latex] (-0.1,0)--(1.15,0);
\draw[-latex] (0,-0.1)--(0,1.15);
\draw (-0.03,1)--(0.03,1);
\draw (1,-0.03)--(1,0.03);
\node[below,scale=0.8] at (1.15,0) { $\delta$ };
\node[left,scale=0.8] at (0,1.15) { $R$ };
\node[below, scale=0.7] at (1,-0.02) { $1$ };
\node[left, scale=0.7] at (-0.02,1) { $1$ };
\node[below left,scale=0.7] at (0,0) { $0$ };
\draw[BC, thick] (0,0.625)--(0.625,0);
\draw[GV, thick] ( 0.669999999999999999999999999994 , 0.00494990003131706750206798437259 )--
( 0.654999999999999999999999999993 , 0.0161585613259531093399348457035 )--
( 0.639999999999999999999999999993 , 0.0278363307498878260145734002830 )--
( 0.624999999999999999999999999993 , 0.0399839135985505715802625651430 )--
( 0.609999999999999999999999999993 , 0.0526020560606280740209821840899 )--
( 0.594999999999999999999999999992 , 0.0656915485280263749755646987213 )--
( 0.579999999999999999999999999992 , 0.0792532292561917879810544667450 )--
( 0.564999999999999999999999999992 , 0.0932879884212941383758600581126 )--
( 0.549999999999999999999999999991 , 0.107796772628290321216217315986 )--
( 0.534999999999999999999999999991 , 0.122780589932847362619302384522 )--
( 0.519999999999999999999999999991 , 0.138240515450837705566054629537 )--
( 0.504999999999999999999999999991 , 0.154177697642036285566680778240 )--
( 0.489999999999999999999999999991 , 0.170593365370270231478236943529 )--
( 0.474999999999999999999999999991 , 0.187488835861262848568627772338 )--
( 0.459999999999999999999999999991 , 0.204865523702628637660612420885 )--
( 0.444999999999999999999999999991 , 0.222724951059021770833101384469 )--
( 0.429999999999999999999999999991 , 0.241068759310759555399198721030 )--
( 0.414999999999999999999999999991 , 0.259898722368232132774104041317 )--
( 0.399999999999999999999999999991 , 0.279216761969586039610573371191 )--
( 0.384999999999999999999999999992 , 0.299024965338899970880459210346 )--
( 0.369999999999999999999999999992 , 0.319325605670908099116003080111 )--
( 0.354999999999999999999999999992 , 0.340121166022487394668616490202 )--
( 0.339999999999999999999999999992 , 0.361414367339200118995901040286 )--
( 0.324999999999999999999999999992 , 0.383208201539182098290394800790 )--
( 0.309999999999999999999999999992 , 0.405505970833608668723931579535 )--
( 0.294999999999999999999999999992 , 0.428311334807307725824394636604 )--
( 0.279999999999999999999999999992 , 0.451628367250511509276576595769 )--
( 0.264999999999999999999999999992 , 0.475461625376212349897177649718 )--
( 0.249999999999999999999999999992 , 0.499816234957184506207735573771 )--
( 0.234999999999999999999999999992 , 0.524697996195980364621827653670 )--
( 0.219999999999999999999999999992 , 0.550113516994984211695354285736 )--
( 0.204999999999999999999999999992 , 0.576070383037114803624399667441 )--
( 0.189999999999999999999999999992 , 0.602577378245238556073106034561 )--
( 0.174999999999999999999999999992 , 0.629644775659473295159776278193 )--
( 0.159999999999999999999999999992 , 0.657284729156467253504680054893 )--
( 0.144999999999999999999999999992 , 0.685511813698294262332515569614 )--
( 0.129999999999999999999999999992 , 0.714343791701591970972807128933 )--
( 0.114999999999999999999999999992 , 0.743802737496352370731245349096 )--
( 0.0999999999999999999999999999916 , 0.773916756697852161971976444050 )--
( 0.0849999999999999999999999999916 , 0.804722754607166117040775398233 )--
( 0.0699999999999999999999999999917 , 0.836271201065466289119765236838 )--
( 0.0549999999999999999999999999917 , 0.868635101921625677720723340723 )--
( 0.0399999999999999999999999999917 , 0.901929207369837427482053419359 )--
( 0.0249999999999999999999999999917 , 0.936360413429565650336482590887 )--
( 0.00999999999999999999999999999170 , 0.972424691598109002129327511475 );
\draw[green!30!black, thick] (0,0.83)--(0.83,0);
\node[scale=0.8] at (0.6, -0.25) { $q = 7^2,\, r = 2$ };
\end{scope}

\begin{scope}[xshift=1.5cm]
\draw[-latex] (-0.1,0)--(1.15,0);
\draw[-latex] (0,-0.1)--(0,1.15);
\draw (-0.03,1)--(0.03,1);
\draw (1,-0.03)--(1,0.03);
\node[below,scale=0.8] at (1.15,0) { $\delta$ };
\node[left,scale=0.8] at (0,1.15) { $R$ };
\node[below, scale=0.7] at (1,-0.02) { $1$ };
\node[left, scale=0.7] at (-0.02,1) { $1$ };
\node[below left,scale=0.7] at (0,0) { $0$ };
\draw[BC, thick] (0, 0.50125)--(0.50125,0);
\draw[GV, thick] ( 0.985000000000000000000000000000 , 0.000184486939418963535243818272466 )--
( 0.969999999999999999999999999999 , 0.000859585263938833472024854610702 )--
( 0.954999999999999999999999999999 , 0.00198468511688698195012898763920 )--
( 0.939999999999999999999999999999 , 0.00355978654653285651683466506090 )--
( 0.924999999999999999999999999999 , 0.00558488960346924954068150547464 )--
( 0.909999999999999999999999999998 , 0.00805999434076383876964083633842 )--
( 0.894999999999999999999999999998 , 0.0109851008141232883410564924527 )--
( 0.879999999999999999999999999998 , 0.0143602090820711806242422313159 )--
( 0.864999999999999999999999999997 , 0.0181853192061412017536806842019 )--
( 0.849999999999999999999999999997 , 0.0224604312510871774687471157961 )--
( 0.834999999999999999999999999997 , 0.0271855452851117543156103542034 )--
( 0.819999999999999999999999999997 , 0.0323606613801157484040700417646 )--
( 0.804999999999999999999999999996 , 0.0379857796119704445015074165422 )--
( 0.789999999999999999999999999996 , 0.0440609000608154279346325065338 )--
( 0.774999999999999999999999999996 , 0.0505860228113848772834494722451 )--
( 0.759999999999999999999999999995 , 0.0575611479533656452277304753398 )--
( 0.744999999999999999999999999995 , 0.0649862755817909177785037334092 )--
( 0.729999999999999999999999999995 , 0.0728614057974737800834665418701 )--
( 0.714999999999999999999999999995 , 0.0811865387074856440221075896810 )--
( 0.699999999999999999999999999994 , 0.0899616744256852258460103965292 )--
( 0.684999999999999999999999999994 , 0.0991868130733046217728497985913 )--
( 0.669999999999999999999999999994 , 0.108861954779600040855797589481 )--
( 0.654999999999999999999999999993 , 0.118987099682575948445098575862 )--
( 0.639999999999999999999999999993 , 0.129562247929792788073309599659 )--
( 0.624999999999999999999999999993 , 0.140587399679270131534227854820 )--
( 0.609999999999999999999999999993 , 0.152062555100499114545820639739 )--
( 0.594999999999999999999999999992 , 0.163987714375580421399660820010 )--
( 0.579999999999999999999999999992 , 0.176362877700506977610109498048 )--
( 0.564999999999999999999999999992 , 0.189188045286614009800111065700 )--
( 0.549999999999999999999999999991 , 0.202463217362223382738694766112 )--
( 0.534999999999999999999999999991 , 0.216188394174514310587135158749 )--
( 0.519999999999999999999999999991 , 0.230363575991658901707154879578 )--
( 0.504999999999999999999999999991 , 0.244988763105268842100907783311 )--
( 0.489999999999999999999999999991 , 0.260063955833209252507934100303 )--
( 0.474999999999999999999999999991 , 0.275589154522847893268050971932 )--
( 0.459999999999999999999999999991 , 0.291564359554823131829299740066 )--
( 0.444999999999999999999999999991 , 0.307989571347433351329250549958 )--
( 0.429999999999999999999999999991 , 0.324864790361775000427934891626 )--
( 0.414999999999999999999999999991 , 0.342190017107787936448923751499 )--
( 0.399999999999999999999999999991 , 0.359965252151407380211625213288 )--
( 0.384999999999999999999999999992 , 0.378190496123074834577926986174 )--
( 0.369999999999999999999999999992 , 0.396865749727930118774859358762 )--
( 0.354999999999999999999999999992 , 0.415991013758099446530351037248 )--
( 0.339999999999999999999999999992 , 0.435566289107619102892494951216 )--
( 0.324999999999999999999999999992 , 0.455591576790703607781477251122 )--
( 0.309999999999999999999999999992 , 0.476066877964300186431565408532 )--
( 0.294999999999999999999999999992 , 0.496992193956196100603589625314 )--
( 0.279999999999999999999999999992 , 0.518367526300404752200591646713 )--
( 0.264999999999999999999999999992 , 0.540192876782216689468722375732 )--
( 0.249999999999999999999999999992 , 0.562468247496267278243160736252 )--
( 0.234999999999999999999999999992 , 0.585193640922412603051477576694 )--
( 0.219999999999999999999999999992 , 0.608369060026398318854286213587 )--
( 0.204999999999999999999999999992 , 0.631994508395727382779130829189 )--
( 0.189999999999999999999999999992 , 0.656069990426614110301666533338 )--
( 0.174999999999999999999999999992 , 0.680595511586964206995644218428 )--
( 0.159999999999999999999999999992 , 0.705571078795794036251739596677 )--
( 0.144999999999999999999999999992 , 0.730996700987025429347224285768 )--
( 0.129999999999999999999999999992 , 0.756872389976871328919055320487 )--
( 0.114999999999999999999999999992 , 0.783198161854953909160967957228 )--
( 0.0999999999999999999999999999916 , 0.809974039331497048268382513667 )--
( 0.0849999999999999999999999999916 , 0.837200055956853655839383142497 )--
( 0.0699999999999999999999999999917 , 0.864876264352776683716292856104 )--
( 0.0549999999999999999999999999917 , 0.893002754135111688963404100097 )--
( 0.0399999999999999999999999999917 , 0.921579697589015724031021411013 )--
( 0.0249999999999999999999999999917 , 0.950607499116167434916293261601 )--
( 0.00999999999999999999999999999170 , 0.980087598431446064110068365472 );
\draw[green!30!black, thick] (0,0.83)--(0.83,0);
\node[scale=0.8] at (0.6, -0.25) { $q = 7^2,\, r \to \infty$ };
\end{scope}
\begin{scope}[xshift=3cm]
\draw[-latex] (-0.1,0)--(1.15,0);
\draw[-latex] (0,-0.1)--(0,1.15);
\draw (-0.03,1)--(0.03,1);
\draw (1,-0.03)--(1,0.03);
\node[below,scale=0.8] at (1.15,0) { $\delta$ };
\node[left,scale=0.8] at (0,1.15) { $R$ };
\node[below, scale=0.7] at (1,-0.02) { $1$ };
\node[left, scale=0.7] at (-0.02,1) { $1$ };
\node[below left,scale=0.7] at (0,0) { $0$ };
\draw[BC, thick] (0, 0.812500000000000000000000000000)--(0.812500000000000000000000000000,0);
\draw[GV, thick] ( 0.699999999999999999999999999994 , 0.00459012343432324350487926144310 )--
( 0.684999999999999999999999999994 , 0.0146651062529533290446099751197 )--
( 0.669999999999999999999999999994 , 0.0252045407995243310217629966239 )--
( 0.654999999999999999999999999993 , 0.0362089401358055226361898821879 )--
( 0.639999999999999999999999999993 , 0.0476788458635265538737463070787 )--
( 0.624999999999999999999999999993 , 0.0596148303352624647295568727452 )--
( 0.609999999999999999999999999993 , 0.0720174990887759280999873998748 )--
( 0.594999999999999999999999999992 , 0.0848874935330763969702512681646 )--
( 0.579999999999999999999999999992 , 0.0982254939187982448798200955416 )--
( 0.564999999999999999999999999992 , 0.112032222630635649407226328550 )--
( 0.549999999999999999999999999991 , 0.126308447845670263711237041390 )--
( 0.534999999999999999999999999991 , 0.141054987608699761400210560149 )--
( 0.519999999999999999999999999991 , 0.156272714384385685334121918279 )--
( 0.504999999999999999999999999991 , 0.171962560156521147859468011402 )--
( 0.489999999999999999999999999991 , 0.188125522157395708633093807594 )--
( 0.474999999999999999999999999991 , 0.204762669325645951712465301052 )--
( 0.459999999999999999999999999991 , 0.221875149609819510168600817080 )--
( 0.444999999999999999999999999991 , 0.239464198258045310848886689590 )--
( 0.429999999999999999999999999991 , 0.257531147262864537806953097561 )--
( 0.414999999999999999999999999991 , 0.276077436165974799245323935689 )--
( 0.399999999999999999999999999991 , 0.295104624472416033103623978762 )--
( 0.384999999999999999999999999992 , 0.314614405980313686827639087609 )--
( 0.369999999999999999999999999992 , 0.334608625404386616020344891022 )--
( 0.354999999999999999999999999992 , 0.355089297764069735129124241593 )--
( 0.339999999999999999999999999992 , 0.376058631127265214290896348830 )--
( 0.324999999999999999999999999992 , 0.397519053458168025414486617792 )--
( 0.309999999999999999999999999992 , 0.419473244526121435749376465141 )--
( 0.294999999999999999999999999992 , 0.441924174111890566118685266200 )--
( 0.279999999999999999999999999992 , 0.464875148127058780394198972655 )--
( 0.264999999999999999999999999992 , 0.488329864784435362068980176441 )--
( 0.249999999999999999999999999992 , 0.512292483687392622292752859918 )--
( 0.234999999999999999999999999992 , 0.536767711744168506579024181174 )--
( 0.219999999999999999999999999992 , 0.561760911317523146705433035167 )--
( 0.204999999999999999999999999992 , 0.587278238246515308943666739591 )--
( 0.189999999999999999999999999992 , 0.613326820750983522721861972480 )--
( 0.174999999999999999999999999992 , 0.639914995480677967162782591936 )--
( 0.159999999999999999999999999992 , 0.667052625398417932537246874474 )--
( 0.144999999999999999999999999992 , 0.694751538196156147953400010116 )--
( 0.129999999999999999999999999992 , 0.723026148209330463510547533695 )--
( 0.114999999999999999999999999992 , 0.751894368923676594181674348047 )--
( 0.0999999999999999999999999999916 , 0.781379008258553149105115596240 )--
( 0.0849999999999999999999999999916 , 0.811510015144895520327218109202 )--
( 0.0699999999999999999999999999917 , 0.842328346240516080150288042876 )--
( 0.0549999999999999999999999999917 , 0.873893246343915381627178287318 )--
( 0.0399999999999999999999999999917 , 0.906297836105531974622231287497 )--
( 0.0249999999999999999999999999917 , 0.939710013209906736151723075885 )--
( 0.00999999999999999999999999999170 , 0.974532262278108756196458285374 );
\draw[green!30!black, thick] (0,0.9)--(0.9,0);
\node[scale=0.8] at (0.6, -0.25) { $q = 11^2,\, r = 2$ };
\end{scope}
\begin{scope}[xshift=4.5cm]
\draw[-latex] (-0.1,0)--(1.15,0);
\draw[-latex] (0,-0.1)--(0,1.15);
\draw (-0.03,1)--(0.03,1);
\draw (1,-0.03)--(1,0.03);
\node[below,scale=0.8] at (1.15,0) { $\delta$ };
\node[left,scale=0.8] at (0,1.15) { $R$ };
\node[below, scale=0.7] at (1,-0.02) { $1$ };
\node[left, scale=0.7] at (-0.02,1) { $1$ };
\node[below left,scale=0.7] at (0,0) { $0$ };
\draw[BC, thick] (0, 0.750625)--(0.750625,0);
\draw[GV, thick] ( 0.985000000000000000000000000000 , 0.000192189428249370972458468397410 )--
( 0.969999999999999999999999999999 , 0.000867269219340869566455130181413 )--
( 0.954999999999999999999999999999 , 0.00199235025076347563446416340702 )--
( 0.939999999999999999999999999999 , 0.00356743256168821092087025838777 )--
( 0.924999999999999999999999999999 , 0.00559251619317150906386910336285 )--
( 0.909999999999999999999999999998 , 0.00806760118827819190451753897848 )--
( 0.894999999999999999999999999998 , 0.0109926875922146386973894389472 )--
( 0.879999999999999999999999999998 , 0.0143677754524731791465344313983 )--
( 0.864999999999999999999999999997 , 0.0181928648189888649267227763742 )--
( 0.849999999999999999999999999997 , 0.0224679557443099153558732733513 )--
( 0.834999999999999999999999999997 , 0.0271930482837832939199020195544 )--
( 0.819999999999999999999999999997 , 0.0323681424957570566747098270302 )--
( 0.804999999999999999999999999996 , 0.0379932384418013250203096882981 )--
( 0.789999999999999999999999999996 , 0.0440683361869499785415845647391 )--
( 0.774999999999999999999999999996 , 0.0505934357999654439971779571939 )--
( 0.759999999999999999999999999995 , 0.0575685373536292806345575839200 )--
( 0.744999999999999999999999999995 , 0.0649936409250616376335958754304 )--
( 0.729999999999999999999999999995 , 0.0728687465960730960367022927048 )--
( 0.714999999999999999999999999995 , 0.0811938544535529163605925350169 )--
( 0.699999999999999999999999999994 , 0.0899689645898983079520300584716 )--
( 0.684999999999999999999999999994 , 0.0991940771034900337648188862411 )--
( 0.669999999999999999999999999994 , 0.108869192099220485004791157990 )--
( 0.654999999999999999999999999993 , 0.118994309689081329025542141002 )--
( 0.639999999999999999999999999993 , 0.129569429992818981747067783805 )--
( 0.624999999999999999999999999993 , 0.140594553138667520775300481525 )--
( 0.609999999999999999999999999993 , 0.152069679264170284599584206029 )--
( 0.594999999999999999999999999992 , 0.163994808517103355742220637571 )--
( 0.579999999999999999999999999992 , 0.176369941056516475544230570490 )--
( 0.564999999999999999999999999992 , 0.189195077053909778729417523076 )--
( 0.549999999999999999999999999991 , 0.202470216694568185346350972492 )--
( 0.534999999999999999999999999991 , 0.216195360179079497092734751967 )--
( 0.519999999999999999999999999991 , 0.230370507725067408078297678314 )--
( 0.504999999999999999999999999991 , 0.244995659569177006940502130137 )--
( 0.489999999999999999999999999991 , 0.260070815969358243159804086559 )--
( 0.474999999999999999999999999991 , 0.275595977207502681379486267766 )--
( 0.459999999999999999999999999991 , 0.291571143592501235535334538408 )--
( 0.444999999999999999999999999991 , 0.307996315463806207116276041387 )--
( 0.429999999999999999999999999991 , 0.324871493195600851462655184576 )--
( 0.414999999999999999999999999991 , 0.342196677201705219445890292022 )--
( 0.399999999999999999999999999991 , 0.359971867941380022904811750071 )--
( 0.384999999999999999999999999992 , 0.378197065926233309413975536456 )--
( 0.369999999999999999999999999992 , 0.396872271728491375218424069456 )--
( 0.354999999999999999999999999992 , 0.415997485990970633575061682919 )--
( 0.339999999999999999999999999992 , 0.435572709439188291358618608768 )--
( 0.324999999999999999999999999992 , 0.455597942896187101945381318336 )--
( 0.309999999999999999999999999992 , 0.476073187300838488555361590937 )--
( 0.294999999999999999999999999992 , 0.496998443730651856105769519209 )--
( 0.279999999999999999999999999992 , 0.518373713430490681918701372629 )--
( 0.264999999999999999999999999992 , 0.540198997849131747972001930685 )--
( 0.249999999999999999999999999992 , 0.562474298686387495278642818665 )--
( 0.234999999999999999999999999992 , 0.585199617954679891478298584998 )--
( 0.219999999999999999999999999992 , 0.608374958060733965379106631739 )--
( 0.204999999999999999999999999992 , 0.632000321915835498821564263047 )--
( 0.189999999999999999999999999992 , 0.656075713087545657195090719502 )--
( 0.174999999999999999999999999992 , 0.680601136013111274109877337850 )--
( 0.159999999999999999999999999992 , 0.705576596307366456682418496248 )--
( 0.144999999999999999999999999992 , 0.731002101220256337863095431237 )--
( 0.129999999999999999999999999992 , 0.756877660340727086115128307447 )--
( 0.114999999999999999999999999992 , 0.783203286725628780921388304374 )--
( 0.0999999999999999999999999999916 , 0.809978998804479855889763650602 )--
( 0.0849999999999999999999999999916 , 0.837204823803645437043841764954 )--
( 0.0699999999999999999999999999917 , 0.864880804426076020871606820247 )--
( 0.0549999999999999999999999999917 , 0.893007013395708037739507793219 )--
( 0.0399999999999999999999999999917 , 0.921583590523249420988381353458 )--
( 0.0249999999999999999999999999917 , 0.950610863983720864451912600853 )--
( 0.00999999999999999999999999999170 , 0.980090002116635375493798561348 );
\draw[green!30!black, thick] (0,0.9)--(0.9,0);
\node[scale=0.8] at (0.6, -0.25) { $q = 11^2,\, r \to\infty$ };
\end{scope}
\begin{scope}[xshift=5cm, yshift=1cm, yscale=0.15, xscale=0.15]
\draw[BC, thick] (0,0)--(1,0);
\node[right,scale=0.8] at (1,0) { \cite[Thm.~4]{BC2024} };
\draw[GV, thick] (0,1)--(1,1);
\node[right,scale=0.8] at (1,1) { GV bound };
\draw[green!30!black, thick] (0,2)--(1,2);
\node[right,scale=0.8] at (1,2) { \Cref{thm:asymptotic}};
\end{scope}
\end{tikzpicture}
\caption{Comparison between GV bound, Theorem 4 of \cite{BC2024} (in the corrected form of \cref{eq:BCbound}), and \Cref{thm:asymptotic}}
\label{fig:GVevencase}
\end{figure}  

\subsubsection{The case $q$ is not a square}

We now assume that $q=p^m$, where $m>$ is odd, and $p$ is a prime number. Further, we choose $j$ and $k$ positive integers such that $i+j=m$ and $\gcd(i,j)=1$. Define for any positive integer $a$ the polynomial $\mathrm{Tr}_j(T)=T+T^p+\cdots+T^{p^{a-1}}$. Now, let $\mathcal F=(F_1,F_2,\dots)$ be the tower over $\Fq$ studied in \cite{BBGS15}. More precisely, one has $F_1=\fq(x_1)$ and $F_i=F_{i-1}(x_i)$, where $x_i$ satisfies
$$\mathrm{Tr}_j\left(\frac{x_i}{x_{i-1}^{p^k}}\right)+\mathrm{Tr}_k\left(\frac{x_i^{q^j}}{x_{i-1}}\right)=1.$$
We cite the following facts from the literature:
\begin{enumerate}
\item From \cite[Cor.~3.2]{BBGS15}: for all $i$ and any nonzero $\alpha \in \fq$, the zero of $x_1-\alpha$ splits completely in the extension $F_i/F_1$. In particular, $N_1(F_i) \ge (q-1)[F_i:F_1].$ 
\item From See \cite[Prop.~2.6]{BBGS15} and the discussion in \cite{BBGS15} directly after Corollary 3.2: for all $i$ only the zero and pole of $x_1$, denoted by $\rl_0$ and $\rl_\infty$ are ramified in $F_i/F_1$. 
\item From \cite[Eq.~25]{BBGS15}: for any place $\pl$ of $F_i$ lying above $\rl_0$ it holds that 
$$d_{\pl|\rl_0} \le b_0(e_{\pl|\rl_0}-1), \quad \text{where} \quad b_0=\frac{p^m-1}{p^k-1}+1.$$
\item From \cite[Eq.~24]{BBGS15}: for any place $\pl$ of $F_i$ lying above $\rl_\infty$ it holds that 
$$d_{\pl|\rl_\infty} \le b_\infty(e_{\pl|\rl_\infty}-1), \quad \text{where} \quad b_\infty=\frac{p^m-1}{p^j-1}+1.$$
\item From the proof of the Main Claim on page 14 in \cite{BBGS15}: any place $\pl$ of $F_i$ that is a zero of $x_i$ is also a zero of $x_1$. Moreover: $\nu_\pl(x_1)=1$ for any such place $\pl$.
\item From \cite[Lem.~2.4]{CCH2021}: for any $i \ge 2$, the extension degree of $F_{i}/F_{i-1}$ is $p^{m-1}.$ Note that in \cite{CCH2021} the language of curves and maps between curves is used rather than function fields and extensions of function fields.
\end{enumerate}

As shown in \cite{BBGS15}, the first four of these facts are enough to deduce that the limit $\lambda$ of the tower $\mathcal F$ described above satisfies the bound
$$\lambda \ge 2\left( \frac{1}{p^j-1}+\frac{1}{p^k-1} \right)^{-1}.$$
To maximize this bound for a given $q=p^m$ with $m$ odd, one can choose $j=\lfloor m/2 \rfloor$ and $k=\lceil m/2 \rceil.$ In this case we obtain that $\lambda \ge H(q)$, just as in equation \eqref{eq:non-square bound}.
With these facts at our disposal, we prove the following result.

\begin{theorem}\label{thm:asymptotic4}
Let $q=p^m$, where $m>1$ is odd, and $p$ is a prime number. For all $R$ and $\delta$ in the interval $(0,1)$ such that 
\[R+\delta < 1-\frac{1}{H(q)},\]
there exists an asymptotic family of $\mathbb{F}_{q^r}$-linearized AG codes with rate at least $R$ and relative minimum distance at least $\delta$.
\end{theorem}
\begin{proof}
We take the tower $\mathcal F=(F_1,F_2,\dots)$ defined in the beginning of this subsection with $j=\lfloor m/2 \rfloor$ and $k=\lceil m/2 \rceil$. To construct the codes, we choose for the function $x$, simply $x:=x_1$ for all $i$. First of all, any zero $\pl$ of $x_i$ in $F_i$ satisfies that $\nu_\pl(x_1)=1$. We claim that there exists a choice of $\pl$ such that $\pl$ is a rational place. We will denote this place by $\pl^{(i)}$. Once we have proved this claim, \Cref{lem:H11} applies, and we can construct sum-rank metric codes using the indicated choice of $x$. To prove the claim, first note that using \cite[Lem.~2.4]{CCH2021}, which we already cited when listing some facts on the tower $\mathcal F$, we know that $[F_i:F_{i-1}]=p^{m-1}$. Considering the pyramid of function fields the tower $\mathcal F$ gives rise to, see \cite[Fig.~5]{BBGS15} for an illustration, we deduce that for any $i \ge 2$ one has $[F_{i-1}:\fq(x_{i-1})]=[F_i:\fq(x_{i-1},x_i)]=(p^{m-1})^{i-1}$ and $[F_i:F_{i-1}]=[\fq(x_{i-1},x_i):\fq(x_{i-1})]=p^{m-1}$. In particular, the fields $F_i$ and $\fq(x_{i-1},x_i)$ are linearly disjoint over $\fq(x_{i-1})$. From \cite[Prop.~2.6 and Figure 2]{BBGS15}, we see that there exists exactly one $\fq$-rational place of $F_2$ that is a common zero of $x_1$ and $x_2$. In the notation of \cite{BBGS15}, it is the place $P_\gamma$ where $\gamma=0$, but we will denote it by $\pl_0^{(1,2)}$. Similarly, there exists a rational place $\pl_0^{(i-1,i)}$ of $\fq(x_{i-1},x_i)$ that is a common zero of $x_{i-1}$ and $x_i$. We claim that there is a place $\pl^{(i)}$ of $F_i$ lying above all the places $\pl_0^{(1,2)}, \dots, \pl_0^{(i-1,i)}$, see \Cref{fig:constructionP} for an illustration, where the notation $[x_j=0]$ is used to denote the zero of $x_j$ in $\fq(x_j)$. The ramification indices indicated in this figure in the bottom layer of the pyramid are from \cite[Prop.~2.8.]{BBGS15}. The ones in the higher layers of the pyramid follow immediately using Abhyankar's lemma \cite[Thm.~3.9.1]{Stich2}.
\begin{figure}[h]
\begin{center}
\scalebox{1.2}{
\makebox[\width][c]{
\def\objectstyle{\scriptstyle}
\xymatrix@!=2.3pc@dr{
&&&&&\\
&\pl^{(i)}\ar@{-}|{e=q^k} [r] & \ar@{.}[r] &\pl_0^{(i-1,i)}\ar@{-}|{e=q^k} [r]& [x_{i}=0]\\
&\pl^{(i-1)}\ar@{-}|{e=q^k}[r] \ar@{-}|{e=1}[u]&\ar@{-}|{e=1}[u] \ar@{.}[r]& [x_{i-1}=0] \ar@{-}|{e=1}[u]\\
&\pl^{(2)}=\pl_0^{(1,2)}\ar@{-}|{e=q^k}[r] \ar@{.}[u]& [x_2=0] \ar@{.}[u]\\
&[x_{1}=0] \ar@{-}|{e=1}[u] \\
}
}
}
\end{center}
\caption{Construction of the place $\pl^{(i)}$.\label{fig:constructionP}}
\end{figure}
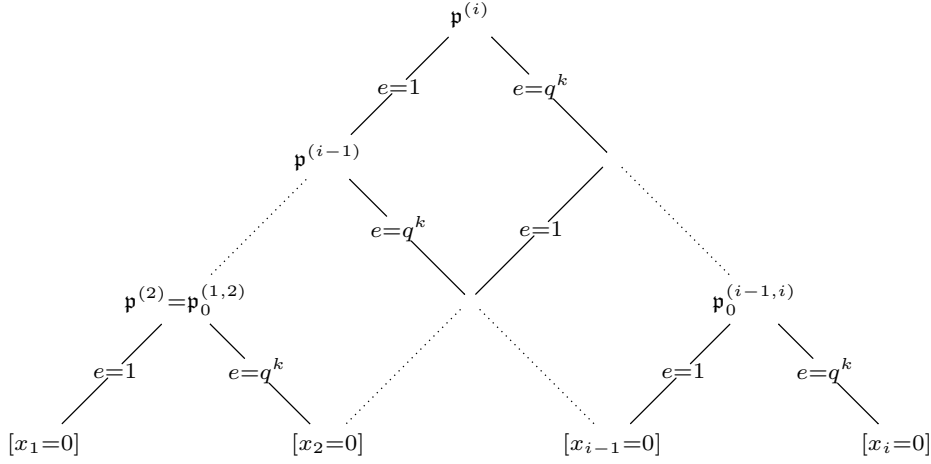
First of all, we can choose $\pl^{(2)}:=\pl_0^{(1,2)}$. Since $F_2$ and $\fq(x_2,x_3)$ are linearly disjoint over $\fq(x_2)$ and have compositum $F_3$, there exists a place $\pl$ of $F_3$ lying above both $\pl^{(2)}$ and $\pl_0^{(2,3)}$. This is a well-known fact proven in the language of function fields in \cite[Lem.~2.1.3]{Wulftange}. Since after completion at $\pl$, the extension degree of $F_1/\fq(x_1)$ becomes $q^k$ and that of $\fq(x_2,x_3)/\fq(x_2)$ becomes $1$, we see that after completion at $\pl$ the extension degree of $F_3/\fq(x_2)$ becomes $q^k$. Since we already know that $e_{\pl|\pl_0^{(2,3)}}=q^k$ using Abhyankar's lemma, we conclude that $f_{\pl|\pl_0^{(2,3)}}=1$. In particular, $\pl$ is $\fq$-rational and we can set $\pl^{(3)}:=\pl$. The place $\pl^{(i)}$ can be constructed similarly using induction on $i$ as a place lying above both $\pl^{(i-1)}$ and $\pl_0^{(i-1,i)}$. Showing that $\pl^{(i)}$ is in fact $\fq$-rational is then done in a very similar way as what we just did for $\pl^{(3)}$.

Now that we have chosen $x$ and shown that \Cref{lem:H11} applies, we return to the construction of codes. We choose as evaluation places all $(q-1)[F_i:F_1]$ rational places of $F_i$ lying above the $q-1$ zeroes of $x_1^{q-1}-1$ in $F_1$. In this way, we can create $\Fqr$-linear sum-rank metric codes $(n_r,k_r,d)$ of $\Fqr$-length $n_r=r(q-1)[F_i:F_1]$, $\Fqr$-dimension
$0 \le k_r \le n_r$ and where $d$ as before satisfies the inequality from \Cref{cor:singdefisotrivial}. 
Now define
$$A_0:=\sum_{\substack{\pl \in \mathbb{P}_{F_i} \\ x_1(\pl)=0}}\deg_{F_i}(\pl) \quad \text{and} \quad A_\infty:=\sum_{\substack{\pl \in \mathbb{P}_{F_i} \\ x_1(\pl)=\infty}}\deg_{F_i}(\pl).$$
Then 
$$\sum_{\pl \in \mathbb{P}_{F_i}\}} \frac{b_\pl-1}{b_\pl} \deg_{F_i}(\pl) \le \frac{r-1}{r}(A_0+A_\infty).$$
Note that using the Riemann--Hurwitz formula, we have
\begin{align*}
2g_{F_i}-2 &= -2[F_i:F_1]+\sum_{\substack{\pl \in \mathbb{P}_{F_i} \\ x_1(\pl)=0}} d_{\pl|\rl_0}\deg_{F_i}(\pl) + \sum_{\substack{\pl \in \mathbb{P}_{F_i} \\ x_1(\pl)=\infty}} d_{\pl|\rl_\infty}\deg_{F_i}(\pl)\\
 & \le -2[F_i:F_1]+\sum_{\substack{\pl \in \mathbb{P}_{F_i} \\ x_1(\pl)=0}} b_0(e_{\pl|\rl_0}-1)\deg_{F_i}(\pl) + \sum_{\substack{\pl \in \mathbb{P}_{F_i} \\ x_1(\pl)=\infty}} b_\infty (e_{\pl|\rl_\infty}-1)\deg_{F_i}(\pl)\\
 & = -2[F_i:F_1]+ \left(b_0[F_i:F_1]-b_0A_0\right)+\left(b_\infty[F_i:F_1]-b_\infty A_\infty\right).
\end{align*}
Then equation \eqref{eq:boundisotrivialcase} gives
\begin{align*}
d+k_r &\ge n_r -\left(r[F_i:F_1]\left(-1+\frac{b_0}{2}+\frac{b_\infty}{2}\right)- \frac{r}2\left(b_0A_0+b_\infty A_\infty\right)+\frac{r-1}2\left(A_0+A_\infty\right)\right)\\
&\ge n_r -\left(r[F_i:F_1]\left(-1+\frac{b_0}{2}+\frac{b_\infty}{2}\right)- \frac{r-1}2\left((b_0-1)A_0+(b_\infty-1) A_\infty\right)\right)\\
&\ge n_r -r[F_i:F_1]\left(-1+\frac{b_0}{2}+\frac{b_\infty}{2}\right).
\end{align*}
In the last inequality, we used that $b_0 \ge 1$ and $b_\infty \ge 1$. Dividing by $n_r=r(q-1)[F_i:F_1]$ and letting $i$ tend to infinity, the theorem follows.
\end{proof}

\begin{figure}
\centering
\begin{tikzpicture}[scale=3]
\begin{scope}[]
\draw[-latex] (-0.1,0)--(1.15,0);
\draw[-latex] (0,-0.1)--(0,1.15);
\draw (-0.03,1)--(0.03,1);
\draw (1,-0.03)--(1,0.03);
\node[below,scale=0.8] at (1.15,0) { $\delta$ };
\node[left,scale=0.8] at (0,1.15) { $R$ };
\node[below, scale=0.7] at (1,-0.02) { $1$ };
\node[left, scale=0.7] at (-0.02,1) { $1$ };
\node[below left,scale=0.7] at (0,0) { $0$ };
\draw[GV, thick] ( 0.699999999999999999999999999994 , 0.00517946961492099911359318423657 )--
( 0.684999999999999999999999999994 , 0.0152487270073131069998415855306 )--
( 0.669999999999999999999999999994 , 0.0257823387817933405656588895419 )--
( 0.654999999999999999999999999993 , 0.0367808145441819219061107281489 )--
( 0.639999999999999999999999999993 , 0.0482446922480162252269140507387 )--
( 0.624999999999999999999999999993 , 0.0601745403905434195557175533520 )--
( 0.609999999999999999999999999993 , 0.0725709604306651484784629830675 )--
( 0.594999999999999999999999999992 , 0.0854345894569035662887544747942 )--
( 0.579999999999999999999999999992 , 0.0987661031377712140867664523453 )--
( 0.564999999999999999999999999992 , 0.112566218992028286663373954830 )--
( 0.549999999999999999999999999991 , 0.126835700022368058599801747173 )--
( 0.534999999999999999999999999991 , 0.141575358763294294281361142659 )--
( 0.519999999999999999999999999991 , 0.156786061802606139565012108752 )--
( 0.504999999999999999999999999991 , 0.172468734846317502912417454031 )--
( 0.489999999999999999999999999991 , 0.188624368409429322540404201536 )--
( 0.474999999999999999999999999991 , 0.205254024230280502901320960045 )--
( 0.459999999999999999999999999991 , 0.222358842524915630640385719645 )--
( 0.444999999999999999999999999991 , 0.239940050220916566525445633264 )--
( 0.429999999999999999999999999991 , 0.257998970338613672305661017066 )--
( 0.414999999999999999999999999991 , 0.276537032723049958585355863534 )--
( 0.399999999999999999999999999991 , 0.295555786374545881668473144677 )--
( 0.384999999999999999999999999992 , 0.315056913681918351560068037148 )--
( 0.369999999999999999999999999992 , 0.335042246934013816298664511216 )--
( 0.354999999999999999999999999992 , 0.355513787577233842095293634058 )--
( 0.339999999999999999999999999992 , 0.376473728806085954832235650459 )--
( 0.324999999999999999999999999992 , 0.397924482230164066167636544140 )--
( 0.309999999999999999999999999992 , 0.419868709568068094916676853075 )--
( 0.294999999999999999999999999992 , 0.442309360596322674938033566294 )--
( 0.279999999999999999999999999992 , 0.465249718958116439099922278763 )--
( 0.264999999999999999999999999992 , 0.488693457955349664496225126964 )--
( 0.249999999999999999999999999992 , 0.512644709172590322450183595779 )--
( 0.234999999999999999999999999992 , 0.537108147812663808230924184569 )--
( 0.219999999999999999999999999992 , 0.562089100117820789795678072496 )--
( 0.204999999999999999999999999992 , 0.587593680461808392609081585431 )--
( 0.189999999999999999999999999992 , 0.613628969049262957229218604983 )--
( 0.174999999999999999999999999992 , 0.640203246374831088528860518579 )--
( 0.159999999999999999999999999992 , 0.667326308965087870367056206809 )--
( 0.144999999999999999999999999992 , 0.695009904841460514625392425626 )--
( 0.129999999999999999999999999992 , 0.723268351245407428514537129044 )--
( 0.114999999999999999999999999992 , 0.752119440998647364640732722319 )--
( 0.0999999999999999999999999999916 , 0.781585828387951228872361550492 )--
( 0.0849999999999999999999999999916 , 0.811697260610300245532107985700 )--
( 0.0699999999999999999999999999917 , 0.842494418442255873212590600114 )--
( 0.0549999999999999999999999999917 , 0.874036147613400606768882628326 )--
( 0.0399999999999999999999999999917 , 0.906414941519828987437629970086 )--
( 0.0249999999999999999999999999917 , 0.939797566898741197356781868087 )--
( 0.00999999999999999999999999999170 , 0.974583797262276424313844958618 );
\draw[green!30!black, thick] (0,0.854166666667)--(0.854166666667,0);
\node[scale=0.8] at (0.6, -0.25) { $q = 5^3,\, r = 2$ };
\end{scope}
\begin{scope}[xshift=1.5cm]
\draw[-latex] (-0.1,0)--(1.15,0);
\draw[-latex] (0,-0.1)--(0,1.15);
\draw (-0.03,1)--(0.03,1);
\draw (1,-0.03)--(1,0.03);
\node[below,scale=0.8] at (1.15,0) { $\delta$ };
\node[left,scale=0.8] at (0,1.15) { $R$ };
\node[below, scale=0.7] at (1,-0.02) { $1$ };
\node[left, scale=0.7] at (-0.02,1) { $1$ };
\node[below left,scale=0.7] at (0,0) { $0$ };
\draw[GV, thick] ( 0.969999999999999999999999999999 , 0.000867491094753157790077857511659 )--
( 0.954999999999999999999999999999 , 0.00199257158035363319764956996947 )--
( 0.939999999999999999999999999999 , 0.00356765333683759883279487938100 )--
( 0.924999999999999999999999999999 , 0.00559273640498493499964471214064 )--
( 0.909999999999999999999999999998 , 0.00806782082757038183936930757764 )--
( 0.894999999999999999999999999998 , 0.0109929066494958115221225806681 )--
( 0.879999999999999999999999999998 , 0.0143679939179336486613683358924 )--
( 0.864999999999999999999999999997 , 0.0181930826824825858328608287005 )--
( 0.849999999999999999999999999997 , 0.0224681729953368811366595222132 )--
( 0.834999999999999999999999999997 , 0.0271932649114706846911757709384 )--
( 0.819999999999999999999999999997 , 0.0323683584888390240301387730497 )--
( 0.804999999999999999999999999996 , 0.0379934537885972884192023214872 )--
( 0.789999999999999999999999999996 , 0.0440685508753412936702276145901 )--
( 0.774999999999999999999999999996 , 0.0505936498173702875296218002121 )--
( 0.759999999999999999999999999995 , 0.0575687506869755776305576399539 )--
( 0.744999999999999999999999999995 , 0.0649938535607578370930234007940 )--
( 0.729999999999999999999999999995 , 0.0728689585199765764707340454305 )--
( 0.714999999999999999999999999995 , 0.0811940656509357761534932215041 )--
( 0.699999999999999999999999999994 , 0.0899691750454102641938607179078 )--
( 0.684999999999999999999999999994 , 0.0991942868011181174428369826270 )--
( 0.669999999999999999999999999994 , 0.108869401022245179120093684325 )--
( 0.654999999999999999999999999993 , 0.118994517820028748353601765685 )--
( 0.639999999999999999999999999993 , 0.129569637313408637380806146914 )--
( 0.624999999999999999999999999993 , 0.140594759629755147815424266142 )--
( 0.609999999999999999999999999993 , 0.152069884905685135608817582278 )--
( 0.594999999999999999999999999992 , 0.163995013287979273680088545047 )--
( 0.579999999999999999999999999992 , 0.176370144934615955170908349057 )--
( 0.564999999999999999999999999992 , 0.189195280015940101605892662683 )--
( 0.549999999999999999999999999991 , 0.202470418715988566461562582284 )--
( 0.534999999999999999999999999991 , 0.216195561233998005697512497752 )--
( 0.519999999999999999999999999991 , 0.230370707786126215077114508208 )--
( 0.504999999999999999999999999991 , 0.244995858607424258076525787646 )--
( 0.489999999999999999999999999991 , 0.260071013954104550926741243995 )--
( 0.474999999999999999999999999991 , 0.275596174106159855936339442508 )--
( 0.459999999999999999999999999991 , 0.291571339370400418932812973237 )--
( 0.444999999999999999999999999991 , 0.307996510083992013876663226339 )--
( 0.429999999999999999999999999991 , 0.324871686618597423245747421980 )--
( 0.414999999999999999999999999991 , 0.342196869385249234300302834092 )--
( 0.399999999999999999999999999991 , 0.359972058840114610077935571241 )--
( 0.384999999999999999999999999992 , 0.378197255491355441272213331253 )--
( 0.369999999999999999999999999992 , 0.396872459907343546862552806446 )--
( 0.354999999999999999999999999992 , 0.415997672726565372631226959348 )--
( 0.339999999999999999999999999992 , 0.435572894669651091078506946670 )--
( 0.324999999999999999999999999992 , 0.455598126554099495783053315724 )--
( 0.309999999999999999999999999992 , 0.476073369312457836154806185161 )--
( 0.294999999999999999999999999992 , 0.496998624014977487321004642956 )--
( 0.279999999999999999999999999992 , 0.518373891898136611207314717312 )--
( 0.264999999999999999999999999992 , 0.540199174400953128296148152499 )--
( 0.249999999999999999999999999992 , 0.562474473211789659043983696806 )--
( 0.234999999999999999999999999992 , 0.585199790329512634104493363810 )--
( 0.219999999999999999999999999992 , 0.608375128144635546798845476930 )--
( 0.204999999999999999999999999992 , 0.632000489548833956694497519168 )--
( 0.189999999999999999999999999992 , 0.656075878085638180728154213150 )--
( 0.174999999999999999999999999992 , 0.680601298162406060769585113076 )--
( 0.159999999999999999999999999992 , 0.705576755356150564714390305418 )--
( 0.144999999999999999999999999992 , 0.731002256867981690051045504564 )--
( 0.129999999999999999999999999992 , 0.756877812222255118489217372598 )--
( 0.114999999999999999999999999992 , 0.783203434387870878767344528263 )--
( 0.0999999999999999999999999999916 , 0.809979141670207426507486522282 )--
( 0.0849999999999999999999999999916 , 0.837204961112234950987885043604 )--
( 0.0699999999999999999999999999917 , 0.864880935129259789481987361444 )--
( 0.0549999999999999999999999999917 , 0.893007135955354901423371065114 )--
( 0.0399999999999999999999999999917 , 0.921583702459472836388875785928 )--
( 0.0249999999999999999999999999917 , 0.950610960606069753322388382588 )--
( 0.00999999999999999999999999999170 , 0.980090070864803462491478087185 );
\draw[green!30!black, thick] (0,0.854166666667)--(0.854166666667,0);
\node[scale=0.8] at (0.6, -0.25) { $q = 5^3,\, r \to \infty$ };
\end{scope}
\begin{scope}[xshift=3cm, yshift=1cm, yscale=0.15, xscale=0.15]
\draw[GV, thick] (0,1)--(1,1);
\node[right,scale=0.8] at (1,1) { GV bound };
\draw[green!30!black, thick] (0,2)--(1,2);
\node[right,scale=0.8] at (1,2) { \Cref{thm:asymptotic4}};
\end{scope}
\end{tikzpicture}
\caption{Comparison between the GV bound and \Cref{thm:asymptotic4}}
\label{fig:GVoddcase}
\end{figure}
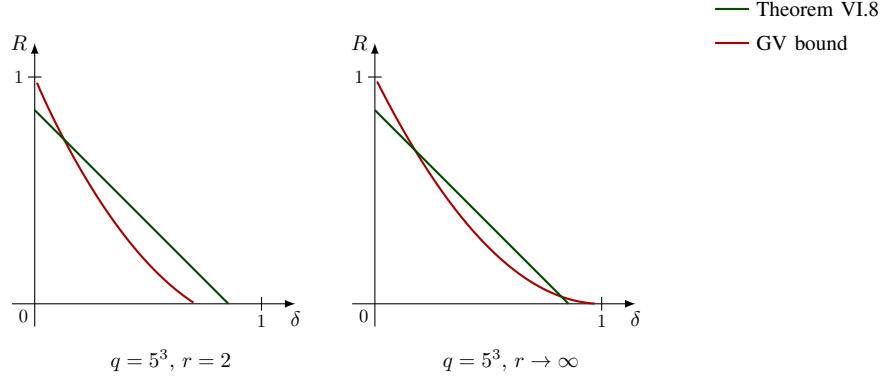

\section{Conclusion and further directions}\label{sec:conclusion}
In this paper, we further investigate AG codes in the sum-rank metric using the construction in \cite{BC2024}. Our first result is to provide a more general formula for the dimension of these codes without assuming injectivity of the corresponding evaluation map (Theorem \ref{thm:exact_dimension}) and a simplified construction under some extra assumptions (\Cref{sec:technicallemmas}). This simplification allows us to provide explicit constructions: in \Cref{sec:explicitconstructions} using rational function fields (obtaining, among other things, explicit MSRD codes from Kummer extensions) and in \Cref{sec:maxfunctfields} using maximal function fields and their very structured Weierstrass semigroups. Lastly, \Cref{sec:asymptotic}, is devoted to asymptotic aspects. A first improvement on the known Gilbert--Varshamov bound for sum-rank metric codes is provided in \Cref{thm:asymptotic3}. This result is followed by further improvements obtained using explicit optimal and good towers, both if $q$ is a square (\Cref{thm:asymptotic}) and if $q$ is not a square and not a prime (\Cref{thm:asymptotic4}). 
\smallskip

We conclude the paper by offering some other possible directions for future research.

The construction of MSRD codes using Kummer extensions (\Cref{constr:Kummer}) is theoretically quite different from the one of linearized Reed--Solomon codes. To emphasize this fact, we provide an explicit description of this code in \Cref{ex:MRD_Kummer} in the case $s=1$, that is, when it is an MRD code. It would be interesting to understand if our MSRD codes from Kummer extensions are equivalent to linearized Reed--Solomon ($s\geq 1$) codes or generalized Gabidulin codes ($s=1$).

With our choice of extension of function fields (in the rational function field case) and our choice of $x$, the obtained codes in Section \ref{sec:explicitconstructions} have length at most $qr^2$. However, a rational function field has $q+1$ rational places, so one could aim to obtain even longer codes. Constructions of longer MSRD codes are known, see for example \cite{MP2023, NSZ2023}, but there the length is increased by adding coordinates in the Hamming metric. 
It would be interesting to understand if considering different extensions of function fields, or different choices of $x$, one could obtain MSRD codes of length $(q+1)r^2$.

Lastly, the construction of sum-rank metric codes from maximal function fields (\Cref{sec:maxfunctfields}) uses for the first time the tool of Weierstrass semigroups, which for the latter class of function fields is very structured. There is, however, a large variety of algebraic function fields for which Weierstrass semigroups are fully determined, see for example \cite{BM2018, BMZ2021, BLM2021, BMB2026}. We believe it could be interesting to see if considering these explicit function fields could lead to more refined constructions.

\section*{Acknowledgments}
This work was supported by a research grant (VIL“52303”) from Villum Fonden. EB is also supported by the grant ANR-22-CPJ2-0047-01.

\providecommand{\bysame}{\leavevmode\hbox to3em{\hrulefill}\thinspace}
\providecommand{\MR}{\relax\ifhmode\unskip\space\fi MR }
\providecommand{\MRhref}[2]{%
  \href{http://www.ams.org/mathscinet-getitem?mr=#1}{#2}
}
\providecommand{\href}[2]{#2}

\end{document}